\documentclass[journal]{IEEEtran}
\usepackage{xcolor,soul,framed}
\usepackage{amsfonts}
\usepackage{textcomp}
\usepackage{ntheorem}
\colorlet{shadecolor}{yellow}
\usepackage[linesnumbered,ruled]{algorithm2e}
\usepackage[pdftex]{graphicx}
\usepackage[cmex10]{amsmath}
\usepackage{array}
\usepackage{eqparbox}
\usepackage{url}
\usepackage{enumerate}   
\usepackage{booktabs}    
\usepackage[figuresright]{rotating}
\usepackage{makecell}

\usepackage{amssymb}   
\usepackage{pifont}

\usepackage{multirow}   
\newtheorem{theorem}{Theorem}
\newtheorem{theorem1}{Theorem}

\newtheorem{assumption}{Assumption}
\newtheorem{definition}{Definition}
\newtheorem{remark}{Remark}
\newtheorem{corollary}{Corollary}
\newtheorem{proposition}{Proposition}
\newtheorem{proposition1}{Proposition}

\usepackage{booktabs}
\usepackage{colortbl}
\usepackage{xcolor}
\usepackage{makecell}
\usepackage{pifont}
\definecolor{oursblue}{RGB}{219, 232, 252}   
\newcommand{\yes}{\ding{51}}                  
\newcommand{\no}{\textcolor{gray}{N/A}}       
\newcommand{\app}{\textit{Approx.}}           

\begin{document}
\makeatletter
\def\@textbottom{\vskip \z@ \@plus 1pt}
\makeatother
\bstctlcite{IEEEexample:BSTcontrol}

\title{Privacy as Commodity: MFG-RegretNet for Large-Scale
       Privacy Trading in Federated Learning}

\author{Kangkang Sun,
        Jianhua Li,~\IEEEmembership{Senior Member,~IEEE,}
        Xiuzhen Chen,~\IEEEmembership{Member,~IEEE,}\\
        Weizhi Meng,~\IEEEmembership{Senior Member,~IEEE,}
        Minyi Guo,~\IEEEmembership{Fellow,~IEEE}
  \thanks{Kangkang Sun, Jianhua Li, Xiuzhen Chen and Minyi Guo are
          with the Shanghai Key Laboratory of Integrated Administration
          Technologies for Information Security, School of Computer Science,
          Shanghai Jiao Tong University, Shanghai 200240, China
          (e-mail: szpsunkk@sjtu.edu.cn; lijh888@sjtu.edu.cn;
          xzchen@sjtu.edu.cn; guo-my@cs.sjtu.edu.cn).}
  \thanks{Weizhi Meng is with the School of Computing and Communications, Lancaster University, Lancaster LA1 4YW, United Kingdom, Email: weizhi.meng@ieee.org.}
  \thanks{\textit{Corresponding author: Jianhua Li}.}
  \thanks{This work has been submitted to the IEEE for possible publication. Copyright may be transferred without notice, after which this version may no longer be accessible.}
}

\maketitle

\begin{abstract}
Federated Learning (FL) has emerged as a prominent paradigm for
privacy-preserving distributed machine learning, yet two fundamental
challenges hinder its large-scale adoption.
First, gradient inversion attacks can reconstruct sensitive training
data from uploaded model updates, so privacy risk persists even when
raw data remain local.
Second, without adequate monetary compensation, rational clients have
little incentive to contribute high-quality gradients, limiting
participation at scale.
To address these challenges, a \textit{privacy trading
market} is developed in which clients sell their differential privacy
budgets as a commodity and receive explicit economic
compensation for privacy sacrifice.
This market is formalized as a \textit{Privacy Auction Game} (PAG),
and the existence of a Bayesian Nash Equilibrium is established under
dominant-strategy incentive compatibility (DSIC), individual
rationality (IR), and budget feasibility.
To overcome the NP-hard, high-dimensional Nash Equilibrium
computation at scale, \textit{MFG-RegretNet} is introduced as a
deep-learning-based auction mechanism that combines mean-field game
(MFG) approximation with differentiable mechanism design.
The MFG reduction lowers per-round computational complexity from
$\mathcal{O}(N^2 \log N)$ to $\mathcal{O}(N)$ while incurring only
an $\mathcal{O}(N^{-1/2})$ equilibrium approximation gap.
Extensive experiments on MNIST and CIFAR-10 demonstrate that
MFG-RegretNet outperforms state-of-the-art baselines in incentive
compatibility, auction revenue, and social welfare, while maintaining
competitive downstream FL model accuracy.
\end{abstract}

\begin{IEEEkeywords}
Federated Learning, Privacy Auction, Privacy Trading,
Mean-Field Game, Differential Privacy
\end{IEEEkeywords}

\IEEEpeerreviewmaketitle

\section{Introduction}
\label{sec:intro}

\IEEEPARstart{F}{ederated} Learning (FL)~\cite{mcmahan} has emerged as
the predominant paradigm for privacy-aware distributed machine
learning, keeping raw data on local devices while sharing only model
gradients with a central server.
Yet a fundamental tension persists: \emph{service providers require
high-quality data to build accurate models, while users are
increasingly reluctant to disclose sensitive information without
adequate compensation or privacy guarantees.}
This tension is exacerbated by two structural deficiencies that
undermine practical FL adoption at scale.

\textbf{First, FL does not eliminate privacy risk.} Model gradients
themselves constitute a serious privacy vulnerability: gradient
inversion attacks~\cite{yang2023gradient} can reconstruct original
training samples from uploaded gradients with high fidelity.
Differentially private FL (DP-FL)~\cite{geyer2017differentially}
mitigates this risk by injecting calibrated noise into gradients before
upload, but this introduces a new dilemma: different clients have
heterogeneous privacy \emph{valuations}: some are willing to upload
relatively clean gradients ($\epsilon$ large) for higher compensation,
while others require strong protection ($\epsilon$ small) and therefore
contribute noisier, less useful gradients. A one-size-fits-all privacy
policy fails to respect this heterogeneity, resulting in either privacy
harm or avoidable accuracy loss.

\textbf{Second, FL lacks effective incentive mechanisms for
privacy-sensitive participants.} In the absence of monetary
compensation, rational users have little motivation to expend local
computation, communication bandwidth, and crucially privacy
budget. Existing Federated Learning Market (FLM) approaches trade
raw datasets, trained model weights, or FL training
capacity~\cite{liMartFLEnablingUtilityDriven2023,
zheng2022fl,
maiAutomaticDoubleAuctionMechanism2022a}, but
\textbf{none treats privacy itself as the traded commodity.}
  The core insight is that the differential privacy budget
  $\epsilon_i$, rather than data, models, or computation, which
  should be the \emph{primary object of trade}, with users receiving
  compensation proportional to the privacy they sacrifice.

\textbf{Three Key Technical Challenges.}
Realising this vision raises three fundamental challenges that prior
work has not jointly addressed.

\textbf{Challenge~1: Privacy Market Design.}
How should the privacy trading market be formally structured? A
rigorous market design must specify: (i) how users express their
heterogeneous privacy valuations; (ii) how a data buyer (e.g., an FL
server) specifies its accuracy requirements and budget; and (iii) how
a trusted broker/auctioneer allocates privacy budgets and determines
payments. Critically, the mechanism must satisfy
\emph{Dominant-Strategy Incentive Compatibility} (DSIC), so that
truthful reporting of privacy valuations is each user's optimal
strategy, and \emph{Individual Rationality} (IR), so that
participation is always at least as good as abstaining. Without these
guarantees, rational users will strategically misreport valuations,
corrupting both market efficiency and FL model quality.

\textbf{Challenge~2: Scalability of Auction Equilibrium.}
Even with a well-defined mechanism, computing the Nash Equilibrium
(NE) of a privacy auction game among $N$ participants is intrinsically
difficult. Each user's optimal bid depends on the strategies of all
other $N-1$ participants, so the NE problem has dimension
$\mathcal{O}(N^{2})$ and is NP-hard in general. For the large-scale
FL settings of interest ($N$ in the hundreds to thousands), exact NE
computation is computationally intractable. \emph{How can
approximate the auction equilibrium efficiently without losing
theoretical guarantees?}

\textbf{Challenge~3: Near-Optimal Mechanism Design under Unknown
Valuation Distributions.}
Challenge~2 reduces the equilibrium problem to $\mathcal{O}(N)$, but
does \emph{not} specify how to find a revenue-maximising allocation
and payment rule. Classical approaches (e.g., Myerson's optimal
auction theory~\cite{myerson1981optimal}) yield closed-form solutions
only under restrictive distributional assumptions. In FL, privacy
valuations are heterogeneous, correlated across rounds, and drawn from
an \emph{unknown} distribution. \emph{How can a mechanism be learned
that is approximately revenue-optimal and DSIC-compliant, without
assuming any closed-form prior over valuations?}

\textbf{Approach and Contributions.}
The three challenges above are addressed through a unified framework
combining a novel \emph{privacy trading market}, a
\emph{mean-field game (MFG) approximation}, and a
\emph{deep learning-based automatic auction solver}, termed
\textbf{MFG-RegretNet}. Inspired by the foundational privacy auction
work~\cite{ghosh2011selling} and
recent progress in differentiable mechanism
design~\cite{dutting2024Optimal}, this framework makes the following
specific contributions.

\begin{itemize}
  \item \textbf{\textit{Privacy Trading Market (PTM).}}
        A formal privacy trading market for FL is established, in which
        each FL client treats its personalised differential privacy
        budget $\epsilon_i$ as a tradeable commodity. Clients bid for
        compensation commensurate with their privacy cost
        $c(v_i,\epsilon_i)=v_i\cdot\epsilon_i$, and the FL server
        purchases privacy budgets subject to a total expenditure
        constraint $B$. To the best of current knowledge, this is the \emph{first}
        work to formalise privacy itself, rather than data, models,
        or computation, as the primary unit of trade in a federated
        learning market.

  \item \textbf{\textit{Privacy Auction Game (PAG) with Theoretical
        Guarantees.}}
        The interaction among clients and the server is modeled as a
        \emph{Privacy Auction Game} (PAG) and design a procurement
        mechanism satisfying DSIC, IR, Dominant-Strategy Truthfulness
        (DT), and Budget Feasibility (BF). It is proven
        (Theorem~\ref{thm:pag-equilibrium}) that under these
        properties, truthful reporting of privacy valuations
        constitutes a Bayesian-Nash Equilibrium of the PAG, providing
        a rigorous incentive foundation for the proposed market.

  \item \textbf{\textit{Mean-Field Game Approximation.}}
        To overcome the $\mathcal{O}(N^2)$ scalability barrier, the
        population of FL clients is modeled as a
        \emph{mean-field flow}
        $\mu_t\in\mathcal{P}(\mathcal{V}\times\mathcal{E})$ and
        reduce the multi-agent Nash problem to a coupled
        Hamilton-Jacobi-Bellman / Fokker-Planck system. It is shown
        (Theorem~\ref{thm:mfg_approx}) that the resulting
        Mean-Field Equilibrium $\varepsilon$-approximates the true
        Nash Equilibrium with error $\mathcal{O}(1/\sqrt{N})$,
        enabling tractable computation in large-scale FL deployments
        while retaining rigorous approximation guarantees.

  \item \textbf{\textit{MFG-RegretNet: Deep Learning Auction Solver.}}
        MFG-RegretNet is proposed as a differentiable mechanism that
        addresses Challenge~3 by parameterising allocation and payment
        rules as neural networks conditioned on the scalar mean-field
        statistic $b_{\mathrm{MFG}}$ (addressing Challenge~2).
        Rather than treating $b_{\mathrm{MFG}}$ merely as a context
        feature, a dedicated \emph{MFG alignment regulariser} forces
        learned payments to agree with the mean-field payment
        functional, so that MFG equilibrium structure enters the
        training gradient. Approximate DSIC is enforced by minimising
        ex-post regret via augmented Lagrangian optimisation, without
        any closed-form prior over privacy valuations.
\end{itemize}


The remainder of this paper is organised as follows.
Section~\ref{sec:related} and Section~\ref{sec:prelim} present related work and
background on differential privacy and federated learning, respectively.
Section~\ref{sec:market} and Section~\ref{sec:pac-fl} present the privacy
trading market and auction game, and the PAG-FL framework, respectively.
Section~\ref{sec:mfg} and Section~\ref{sec:regretnet} present the MFG
approximation and MFG-RegretNet, respectively.
Section~\ref{sec:exp} and Section~\ref{sec:conclusion} present experimental
results and concluding remarks, respectively.

\section{Related Work}
\label{sec:related}

\begin{table*}[t]
\centering
\caption{Comparison of Representative Related Works Across Key Design Dimensions}
\label{tab:related}
\resizebox{\textwidth}{!}{
\begin{tabular*}{\textwidth}{@{\extracolsep{\fill}}
  l                              
  c c c c c c                    
@{}}
\toprule
\multirow{2}{*}{\textbf{Work}}
  & \textbf{FL}        & \textbf{Privacy as}  & \textbf{Auction}
  & \textbf{MFG}       & \textbf{Deep}        & \textbf{DSIC}       \\
  & \textbf{Setting}   & \textbf{Commodity}   & \textbf{Mechanism}
  & \textbf{Scalability} & \textbf{Learning}  & \textbf{Guarantee}  \\
\midrule
Ghosh \& Roth~\cite{ghosh2011selling}
  & \no & \yes & \yes & \no & \no & \yes \\
Fleischer \& Lyu~\cite{fleischer2012approximately}
  & \no & \yes & \yes & \no & \no & \app \\
Yang et al.~\cite{yang2024CSRA}
  & \yes & \no & \yes & \no & \no & \yes \\
Zheng et al.~\cite{zheng2022fl}
  & \yes & \no & \yes & \no & \yes & \app \\
Sun et al.~\cite{sun2024reputation}
  & \yes & \no & \no & \yes & \no & \no \\
Hu et al.~\cite{hu2025federated}
  & \yes & \no & \no & \yes & \no & \no \\
D\"{u}tting et al.~\cite{dutting2024Optimal}
  & \no & \no & \yes & \no & \yes & \app \\
\midrule
\rowcolor{oursblue}
\textbf{MFG-RegretNet (Ours)}
  & \yes & \yes & \yes & \yes & \yes & \app \\
\bottomrule
\end{tabular*}}
\vspace{2pt}
{\centering \footnotesize\raggedright
  \yes~$=$ supported;\quad
  \textcolor{gray}{N/A}~$=$ not supported;\quad
  \textit{Approx.}~$=$ approximately satisfied via augmented Lagrangian training.
}
\end{table*}

This study lies at the intersection of four research streams,
summarised in Table~\ref{tab:related}, with extended discussion in Appendix~\ref{app:related}\footnote{Appendix is available at \url{https://github.com/szpsunkk/MFG-RegretNet}}.

\textbf{Privacy-preserving FL.}
Gradient inversion attacks~\cite{yang2023gradient} demonstrate that
an honest-but-curious server can reconstruct raw samples from shared gradients.
The dominant defence is Differential Privacy: DP-SGD and its federated
variant~\cite{geyer2017differentially} inject calibrated Gaussian noise into
gradients before upload, offering $(\epsilon,\delta)$-DP at the cost of
model utility.
Secure aggregation~\cite{chi2023trusted} uses cryptographic masking so the
server observes only the aggregate.
A critical gap persists: all these approaches treat $\epsilon$ as a fixed
system parameter, never as a user-controlled commodity tradeable for compensation.

\textbf{FL incentive mechanisms.}
Incentive design for FL has been studied via Stackelberg
games~\cite{yi2022Stackelberg,chenMultifactorIncentiveMechanism2023},
MFG-based reputation auctions~\cite{sun2024reputation}, matching
games~\cite{singhStable2024}, and DRL-based pricing~\cite{zhan2020learning}.
Auction-based FL (AFL)~\cite{tang2024survey} has emerged as a prominent
sub-field, with CSRA~\cite{yang2024CSRA} providing robust reverse auctions
for DP-FL with truthfulness and budget-feasibility guarantees.
However, all existing AFL works treat privacy as a cost constraint, trading
data, models, or compute capacity rather than the privacy budget itself.

\textbf{Privacy auctions and data markets.}
Ghosh and Roth~\cite{ghosh2011selling} laid the theoretical foundation for
privacy-as-commodity by formulating a reverse auction that elicits individual
privacy valuations.
Fleischer and Lyu~\cite{fleischer2012approximately} extended this to correlated
privacy costs; Ghosh et al.~\cite{ghosh2014buying} addressed unverifiable
valuations.
More recently, Hu et al.~\cite{hu2025federated} proposed MFG-based data pricing
for federated markets; Wu et al.~\cite{wu2026flgcpa} treated privacy as a
truthfully reported resource in FL, which is the closest prior art to our PAG framing.
The key limitation of all prior work is a focus on \emph{static, single-round,
centralised} settings; none addresses the multi-round distributed FL context.

\textbf{Differentiable mechanism design.}
Myerson~\cite{myerson1981optimal} gives revenue-optimal auctions under
restrictive distributional assumptions inapplicable to heterogeneous FL clients.
RegretNet~\cite{dutting2024Optimal} parameterises allocation and payment rules
as neural networks trained to minimise ex-post regret, circumventing distributional
priors; DM-RegretNet~\cite{zheng2022fl} applied this to FL model trading.
Both inherit an $\mathcal{O}(N^2)$ pairwise dependency that makes them
computationally intractable at large $N$.
MFG-RegretNet replaces pairwise coupling with a scalar mean-field statistic,
reducing complexity to $\mathcal{O}(N)$ while operating on privacy budgets
as the traded commodity.

\textbf{Analytical DP-FL selection and pricing.}
Recent analytical approaches address heterogeneous-$\epsilon$ selection
and compensation directly.
Privacy-aware client selection~\cite{li2022privacy} and heterogeneous
DP aggregation~\cite{liu2022gdpfed} optimise selection probabilities
under non-uniform noise variance, but provide no economic mechanism or
incentive-compatibility guarantees.
Jointly-optimal selection-and-allocation mechanisms (JSAM-style,
e.g.,~\cite{ding2021differentially}) derive structural results
(threshold selection, $\epsilon^*_k \propto p_k^{2/3}v_k^{-1/3}$)
under Bayesian priors, but do not scale to unknown distributions or
large $N$.
Our key distinctions are: (i)~\emph{distribution-free} learned
mechanism via regret minimisation; (ii)~\emph{MFG-aware regret}
that accounts for population-level strategic effects;
and (iii)~a formal $\mathcal{O}(N^{-1/2})$ MFE-to-Nash
approximation bound enabling scalable deployment.

\section{Preliminaries}
\label{sec:prelim}

\subsection{Differential Privacy}

We adopt \emph{Personalised Local Differential Privacy}
(PLDP)~\cite{jorgensen2015conservative}, which generalises standard
LDP~\cite{dwork2014algorithmic} to allow each client to hold a
distinct privacy budget.

\begin{definition}[PLDP~\cite{jorgensen2015conservative}]
\label{de:PLDP}
For client $i\in\mathcal{N}$ with $\epsilon_i>0$ and $\delta_i\geq 0$,
a randomised mechanism $\mathcal{M}^i$ is
$(\epsilon_i,\delta_i)$-PLDP if for all adjacent inputs
$\boldsymbol{x},\boldsymbol{x}'\in\mathcal{X}^i$ and all outputs
$\boldsymbol{y}\in\mathcal{Y}^i$:
\begin{equation}
  \Pr[\mathcal{M}^i(\boldsymbol{x})=\boldsymbol{y}]
  \leq e^{\epsilon_i}
  \Pr[\mathcal{M}^i(\boldsymbol{x}')=\boldsymbol{y}] + \delta_i.
\end{equation}
Smaller $\epsilon_i$ implies stronger privacy at the cost of lower
data utility.
\end{definition}

In practice, private gradients are perturbed via the
Gaussian Mechanism~\cite{bun2016concentrated}: setting noise scale
$\sigma_i \geq \tfrac{2\Delta f\sqrt{\ln(1.25/\delta_i)}}{\epsilon_i}$
makes the perturbed gradient $(\epsilon_i,\delta_i)$-DP,
where $\Delta f$ is the $\ell_2$-sensitivity of the gradient function.
This relation ties the allocated budget $\epsilon_i^{\mathrm{out}}$
directly to the noise injected into each client's gradient
(see Algorithm~\ref{alg:mfg-fpa-online}).

\subsection{Federated Learning}

In standard FL, a server $S$ coordinates $N$ clients, each holding
a private dataset $\mathcal{D}_i$.
Client $i$ minimises its local loss
$\mathcal{L}_i(\mathbf{w},\mathcal{D}_i)$ to obtain local model
$\mathbf{w}_i$; the server aggregates:
\begin{equation}
\label{eq:global-gradient}
  \mathbf{w}_G = \sum_{i=1}^{N}\alpha_i\mathbf{w}_i,
  \quad \sum_{i=1}^{N}\alpha_i = 1,
\end{equation}
targeting $\mathbf{w}_G^* = \arg\min_{\mathbf{w}}\sum_i \alpha_i \mathcal{L}_i(\mathbf{w}_i)$.
Under DP-FL, each client perturbs its gradient with Gaussian noise
calibrated to $\epsilon_i^{\mathrm{out}}$ before upload, so the
server can only observe the noisy aggregate.
The $\ell_2$-sensitivity of client $i$'s gradient is:
\begin{equation}
  \Delta f_{\mathbf{w}_i}
  = \max_{\mathcal{D}_i,\mathcal{D}_i^{\prime}}
    \left\|\nabla\mathcal{L}_i(\mathbf{w},\mathcal{D}_i)
           -\nabla\mathcal{L}_i(\mathbf{w},\mathcal{D}_i')\right\|_2,
\end{equation}
where $\mathcal{D}_i'$ differs from $\mathcal{D}_i$ in exactly one
sample.
The allocated budget $\epsilon_i^{\mathrm{out}}$ thus determines
the noise scale $\sigma_i$ and directly affects gradient quality
and downstream model accuracy, making it the natural economic object
to trade in the privacy market.

\section{Privacy Trading Market and Privacy Auction Game}
\label{sec:market}

In this section, we introduce the \emph{privacy trading market},
where privacy as a commodity trades among data owners (DOs) and data
buyers (DBs). We then design a \emph{privacy auction game} to
incentivise participants to share data.

\begin{figure}[t]
  \centering
  \includegraphics[width=0.48\textwidth]{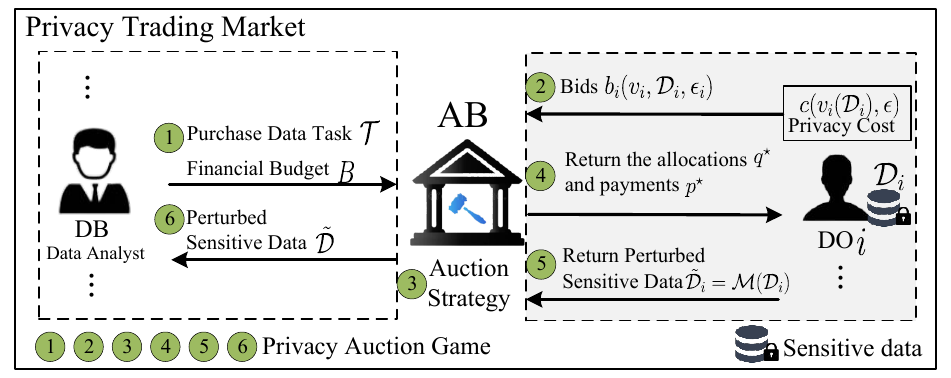}
  \caption{Privacy Trading Market: the six-step auction process
           among DOs, the Auction Broker (AB), and DBs.}
  \label{fig:privacy-market}
\end{figure}

\subsection{Privacy Trading Market}
\label{subsec:ptm}

The privacy trading market comprises three roles:
\textit{Data Owners} (DOs), an \textit{Auction Broker} (AB), and
\textit{Data Buyers} (DBs).
As illustrated in Fig.~\ref{fig:privacy-market}, the auction proceeds
in six steps.
DOs $\mathcal{N}=\{1,\ldots,n\}$ hold local sensitive datasets
$\mathcal{D}=\{\mathcal{D}_1,\ldots,\mathcal{D}_n\}$.
We assume that DOs are \textit{single-minded}: a DO will provide
raw sensitive data only if the privacy preservation it requests is
satisfied~\cite{zhang2020Selling}.
DO $i$ protects her sensitive data through a mechanism
$\mathcal{M}_{\epsilon_i}$ satisfying $(\epsilon_i,\delta_i)$-PLDP.
DBs purchase sensitive data from DOs to complete analysis or model
training tasks. The AB is a credible auctioneer that performs multiple
rounds of auctions. Since a DO may misreport its privacy valuation
$v(\mathcal{D})$, the mechanism must incentivise truthful reporting.

We assume the AB is fully trusted and a reverse auction mechanism is
used~\cite{ghosh2011selling,dandekar2014privacy}.

\begin{assumption}
\label{ass:market}
All buyers agree to use privacy as a common evaluation metric,
making the auction a fair game. Buyers have strong beliefs in the
true value of their data; $v$ represents the user's estimation of
data privacy during the auction. Buyers do not collude with each
other during the auction process.
\end{assumption}

\begin{remark}[Threat Model and Trust Assumptions]
\label{rem:threat}
Our framework operates under the following threat model.
\textbf{(1) Trusted Auction Broker (AB).}
We assume the AB is a semi-trusted, non-colluding third party that
faithfully executes the auction mechanism and does not leak bid
information to other clients.
This is consistent with existing privacy auction
literature~\cite{ghosh2011selling,dandekar2014privacy} and is a
standard assumption in mechanism design.
In practice, the AB can be instantiated as a secure multi-party
computation (SMPC) node or a trusted execution environment (TEE),
as explored in adjacent works.
\textbf{(2) No Client Collusion.}
Clients are assumed to act independently (Assumption~\ref{ass:market}).
Collusion among a subset of clients could allow coordinated
bid manipulation to inflate payments; analysing collusion-resistant
mechanisms is an important direction for future work.
\textbf{(3) Relationship to Secure Aggregation.}
Our mechanism operates at the \emph{auction layer} (determining
which privacy budget each client uses), independently of the
gradient aggregation layer.
It is fully compatible with secure aggregation
protocols~\cite{chi2023trusted} that mask individual gradients
from the server: the AB only observes reported bids $(v_i, \epsilon_i)$,
not raw gradients.
Combining our market mechanism with secure aggregation thus provides
both economic incentives (this work) and cryptographic privacy
guarantees (orthogonal protocol layer).
\textbf{(4) Disclosure of $\epsilon_i$.}
Clients report their declared privacy budget $\epsilon_i$ as part
of the bid.
Under DSIC, truthful reporting is a dominant strategy, so the
disclosed $\epsilon_i$ does not create additional privacy leakage
beyond what is inherent in the auction participation itself.
\end{remark}

\subsection{Privacy Auction Game (PAG)}
\label{subsec:pag}

To incentivise privacy trading, we design a privacy auction game
(Fig.~\ref{fig:privacy-market}).
Each DO $i$ has a valuation function $v_i(\mathcal{D}_i;b_i)$, where
bidding profile $b_i$ is given at privacy budget
$(\epsilon_j,\delta_j)$-PLDP.
DOs report their valuations (possibly untruthfully) and the AB decides
an allocation and payment $p_1,\ldots,p_n$.

For DO $i$, the cost of having data used at privacy level $\epsilon$
is $c(v_i,\epsilon)$, linear in both $v_i$ and $\epsilon$:
$c(v_i,\epsilon)\leq c(v_j,\epsilon)$ if and only if $v_i\leq v_j$.
Specifically: $c_i(\epsilon)=\epsilon\cdot v_i$~\cite{fleischer2012approximately}.

The utility of DO $i$ is:
\begin{equation}
  u_i = p_i - c(v_i,\epsilon_i).
\end{equation}
Summing client utilities defines \emph{participant social welfare}
(client-side economic surplus): monetary compensation net of privacy
disutility. In one auction round,
\begin{equation}
\label{eq:sw-def}
  \mathrm{SW}
  = \sum_{i=1}^{n} u_i
  = \sum_{i=1}^{n}\bigl(p_i - v_i\epsilon_i^{\mathrm{out}}\bigr),
\end{equation}
where $\epsilon_i^{\mathrm{out}}$ is the privacy level actually applied
to client $i$ after allocation (in FL, the noise scale of the uploaded
gradient). Under IR, $p_i\ge c(v_i,\epsilon_i^{\mathrm{out}})$ and each
summand in~Eq.~(\ref{eq:sw-def}) is non-negative. Payments are transfers from
the server's budget $B$, so $\mathrm{SW}$ does not include the server's
utility from model quality; the latter is evaluated via downstream
accuracy (RQ4, Sec.~\ref{sec:exp}). In experiments we evaluate
Eq.~(\ref{eq:sw-def}) using each client's \emph{true} $v_i$ when forming
$c(\cdot)$, so that welfare aligns with realised privacy cost and is
comparable across mechanisms. Across $T_{\mathrm{FL}}$ federated rounds
we report the time average
$\overline{\mathrm{SW}}=\frac{1}{T_{\mathrm{FL}}}\sum_{t=1}^{T_{\mathrm{FL}}}
\mathrm{SW}^{(t)}$, where $(p_i^{(t)},\epsilon_i^{(t),\mathrm{out}})$ denote
round-$t$ outcomes. We write $\mathrm{SW}$ to distinguish welfare from
the type distribution denoted $\mathcal{W}$ in
Theorem~\ref{thm:pag-equilibrium}.

Define the joint valuation domain
$\mathcal{V}=\mathcal{V}_1\times\cdots\times\mathcal{V}_n$, where
$\mathcal{V}_i=\{v_{i1},\ldots,v_{im}\}$ ($m$ = number of items).
Let $\mathcal{V}_{-i}$ and $b_{-i}$ denote the joint domain and bid
vector excluding agent $i$, respectively.

\subsection{Privacy Auction Equilibrium}
\label{subsec:pae}

All bidders aim to maximise $u_i(v_i,b)$. To obtain a
\textit{Nash Equilibrium}, we design a \textit{privacy procurement
mechanism}.

\begin{definition}[Privacy Procurement Mechanism]
\label{de:mechanism}
A privacy procurement mechanism $\mathcal{M}=(q,p)$ consists of:
\begin{itemize}
  \item an allocation function
        $q:\mathcal{V}\rightarrow\{0,1\}^n$;
  \item a payment function
        $p:\mathcal{V}\rightarrow\mathbb{R}^n$.
\end{itemize}
\end{definition}

The mechanism $\mathcal{M}$ must satisfy the following properties:

\begin{itemize}
  \item \textbf{Dominant-Strategy Incentive Compatibility (DSIC):}
        \label{prop:DSIC}
        For each bidder $i\in\mathcal{N}$ and
        $\forall b_{-i}\in\mathcal{V}_{-i}$, truthful reporting of
        valuation $v_i$ is a dominant strategy regardless of others'
        reports:
        \begin{equation}
          u_i\!\left(\cdot|(v_i,b_{-i})\right)
          \geq u_i\!\left(\cdot|(b_i,b_{-i})\right),
          \quad\forall b_i\in\mathcal{V}_i.
        \end{equation}

  \item \textbf{Individual Rationality (IR):}
        \label{prop:IR}
        For each bidder $i\in\mathcal{N}$, participation yields
        non-negative utility:
        \begin{equation}
          p_i(v) \geq c(v_i,\epsilon_i).
        \end{equation}

  \item \textbf{Dominant-Strategy Truthfulness (DT):}
        \label{prop:DT}
        DSIC guarantees truthful reporting of valuations $v_i$.
        DT further requires that each bidder also truthfully reports
        her \emph{privacy budget} $\epsilon_i$, i.e., for any fake
        reported budget $\epsilon'_i \neq \epsilon_i$:
        \begin{equation}
          u_i(v_i,\epsilon_i,b_i;b_{-i})
          \geq u_i(v_i,\epsilon_i,b'_i;b_{-i}).
        \end{equation}
        This is necessary because the cost $c(v_i,\epsilon_i)=v_i\cdot\epsilon_i$
        depends on the reported $\epsilon_i$; a client could
        misreport $\epsilon_i$ (in addition to $v_i$) to
        manipulate its cost and payment.

  \item \textbf{Budget Feasibility (BF):}
        \label{prop:BF}
        \begin{equation}
        \label{eq:BF}
          \sum_{i=1}^n p_i(v_i)\leq B.
        \end{equation}
\end{itemize}

Under DSIC, IR, and DT, the expected revenue given joint distribution
$\mathcal{W}$ of $(v,\epsilon,\delta)$ is:
\begin{equation}
  \mathcal{R}
  := \mathbb{E}_{(v,\epsilon)\sim\mathcal{D}_{v,\epsilon}}
     \!\left[\sum_{i=1}^n p_i(v,\epsilon)\right].
\end{equation}

The optimal privacy auction problem is:
\begin{equation}
\label{eq:P1}
  \mathcal{P}_1:\quad
  \max_{v,\epsilon}\sum_{i=1}^{N}\mathcal{R}_i
  \quad\text{s.t.: DSIC, IR, DT, BF.}
\end{equation}

\begin{definition}[Privacy Auction Equilibrium]
\label{de:PAE}
A strategy profile $q_1,\ldots,q_n$ is a \textit{Bayesian-Nash
Equilibrium} (BNE) if, for all $i\in\mathcal{N}$ and
$v_i\in\mathcal{V}$:
\begin{equation}
  \mathbb{E}[u_i(q,v_i)] \geq \mathbb{E}[u_i(q_{-i},v_i)].
\end{equation}
\end{definition}

\begin{theorem}[Privacy Auction Game Equilibrium]
\label{thm:pag-equilibrium}
In a PAG where privacy valuations $v_i$ of DOs are drawn from joint
distribution $\mathcal{W}$, the cost function is linear
$c(v_i,\epsilon_i)=\epsilon_i\cdot v_i$, and the mechanism
$\mathcal{M}=(q,p)$ satisfies DSIC, IR, DT, and BF, there exists a
Bayesian-Nash Equilibrium in which all rational bidders truthfully
report both their privacy valuations $v_i$ \emph{and} their privacy
budgets $\epsilon_i$:
\begin{equation}
\label{eq:bne}
  \forall i\in\mathcal{N},\forall v_i\in\mathcal{V}_i,\quad
  \mathbb{E}[u_i(q,v_i)] \geq \mathbb{E}[u_i(q_{-i},v_i)].
\end{equation}
\end{theorem}
All the proofs are in the appendix.

\begin{remark}[Scope of DSIC/DT Guarantees]
\label{rem:dsic-scope}
Theorem~\ref{thm:pag-equilibrium} and DT are formally established for the
\emph{PAC mechanism} (Algorithm~\ref{alg:pac}).
\emph{MFG-RegretNet} satisfies only \textbf{approximate} DSIC via
augmented-Lagrangian regret minimisation; exact DSIC for arbitrary neural
payment rules requires monotonicity conditions not verifiable for deep
networks~\cite{dutting2024Optimal}.
All theoretical DSIC/DT/BNE claims apply to PAC; MFG-RegretNet
properties are characterised empirically (RQ1, Section~\ref{sec:exp}).
\end{remark}

\begin{algorithm}[t]
  \small
\caption{PAC Mechanism}
\label{alg:pac}
\KwIn{Privacy valuations $v$, datasets $\mathcal{D}$, budget $B$}
\KwOut{Perturbed data $\mathcal{D}'$, payments $p$}
All participants submit privacy valuation $v$ to AB\;
Sort valuations: $v_1\leq v_2\leq\cdots\leq v_n$\;
Let $k$ be the largest integer such that
$c\!\left(v_k,\,\tfrac{1}{n-k}\right)\leq\tfrac{B}{k}$\;
\tcp{$\epsilon{=}1/(n{-}k)$: equal-share DP budget across $n{-}k$ non-winners; ensures each winner's cost $c(v_k,\epsilon)$ fits within per-client budget $B/k$. This is the tightest uniform $\epsilon$ consistent with DSIC and BF.}
\For{$i=1$ \KwTo $n$}{
  Compute perturbed data:
  $\mathcal{D}'_i \leftarrow \mathcal{D}_i
   + \mathcal{N}(0,\,\Delta_f^2\sigma^2)$\;
  \eIf{$i\leq k$}{
    $p_i \leftarrow
    \min\!\left(\tfrac{B}{k},\,
    c\!\left(v_{k+1},\tfrac{1}{n-k}\right)\right)$\;
  }{
    $p_i \leftarrow 0$\;
  }
}
\Return $\mathcal{D}'$, $p$\;
\end{algorithm}

\begin{remark}[PAC Formal Allocation Rule and DSIC Verification]
\label{rem:pac-formal}
PAC selects the \emph{$k$ lowest-valuation clients} as winners
(monotone allocation: adding a client can only increase $k$, not
decrease it), assigns each winner
$\epsilon_i^{\mathrm{out}} = 1/(n-k)$ (equal-share DP budget), and
pays each winner
$p_i = \min(B/k,\,c(v_{k+1}, 1/(n-k)))$, i.e., a threshold
payment equal to the critical bid at which the marginal client would
be selected.
The three properties follow directly:
\emph{(i) DSIC}: reporting $v_i' > v_i$ cannot increase $p_i$
(threshold payment is determined by $v_{k+1}$, not $v_i$);
misreporting $v_i' < v_i$ either leaves $p_i$ unchanged (if $i$
remains a winner) or yields $p_i = 0$ (if $i$ drops out), so
$u_i(v_i) \geq u_i(v_i')$.
\emph{(ii) IR}: $p_i = \min(B/k, c(v_{k+1},\epsilon)) \geq
c(v_k,\epsilon) \geq c(v_i,\epsilon)$ for $i\leq k$ (since
$v_i \leq v_k \leq v_{k+1}$).
\emph{(iii) BF}: $\sum_i p_i \leq k \cdot B/k = B$.
The choice $\epsilon=1/(n-k)$ is the tightest uniform budget
satisfying $c(v_k,\epsilon)\leq B/k$; it preserves DT because a
client cannot gain by misreporting $\epsilon_i$ given the threshold
structure.
\emph{Complexity}: $\mathcal{O}(N\log N)$ (sorting); the MFG
approximation reduces deploy-time NE complexity from
$\mathcal{O}(N^2\log N)$ to $\mathcal{O}(N)$
(Section~\ref{sec:mfg}).
\end{remark}

\section{Privacy-Aware Auction Federated Learning Mechanism}
\label{sec:pac-fl}

We combine the privacy auction game with FL to incentivise users to
share data for three reasons. First, users are concerned about privacy
risks from local model leakage. Second, without revenue, users have
limited incentive to participate in FL training. Third, a unified
privacy policy ignores user heterogeneity, reducing model accuracy.
We propose the \emph{Privacy Auction Game-based FL} (PAG-FL) framework
(Fig.~\ref{fig:framework}).

\subsection{PAG-FL Framework}
\label{subsec:pagfl}

PAG-FL involves three entities:
\begin{enumerate}
  \item \textit{FL Server (FLS).}
        FLS publishes model training tasks (e.g., image classification,
        text translation) and purchases client gradients via the FL
        architecture.
  \item \textit{FL Broker (FLB).}
        FLB is the auction institution for user data privacy. It
        recruits FLCs eligible for a given task and runs the privacy
        auction to determine allocation and payment.
  \item \textit{FL Clients (FLC).}
        FLCs are users or organisations that provide raw training data.
        Each FLC trains locally, applies a DP mechanism, and uploads
        a private gradient to FLB for auction.
\end{enumerate}

The operational steps are:
\begin{enumerate}
  \item \textit{Publish Task.}  FLS publishes a learning task to FLB,
        which recruits eligible FLCs.
  \item \textit{Registration.}  FLCs register with FLB and declare
        initial privacy requirements $(\epsilon,\delta)$.
  \item \textit{Submit Bid.}  Each FLC $i$ submits a bid
        $b_i=(v_i,d_i,u_i,\epsilon_i)$, where $v_i$ is a valuation,
        $d_i$ is model size, $u_i$ is local model accuracy, and
        $\epsilon_i$ is the privacy budget.
  \item \textit{Privacy Payment.}  FLB runs the auction (PAC
        mechanism, Algorithm~\ref{alg:pac}) and FLS pays selected FLCs.
  \item \textit{Private Model.}  Each selected FLC applies its
        personalised DP mechanism to the local gradient.
  \item \textit{Model Aggregation.}  FLS aggregates received private
        gradients and broadcasts the updated global model.
\end{enumerate}

\begin{figure}[t]
  \centering
  \includegraphics[width=1\linewidth]{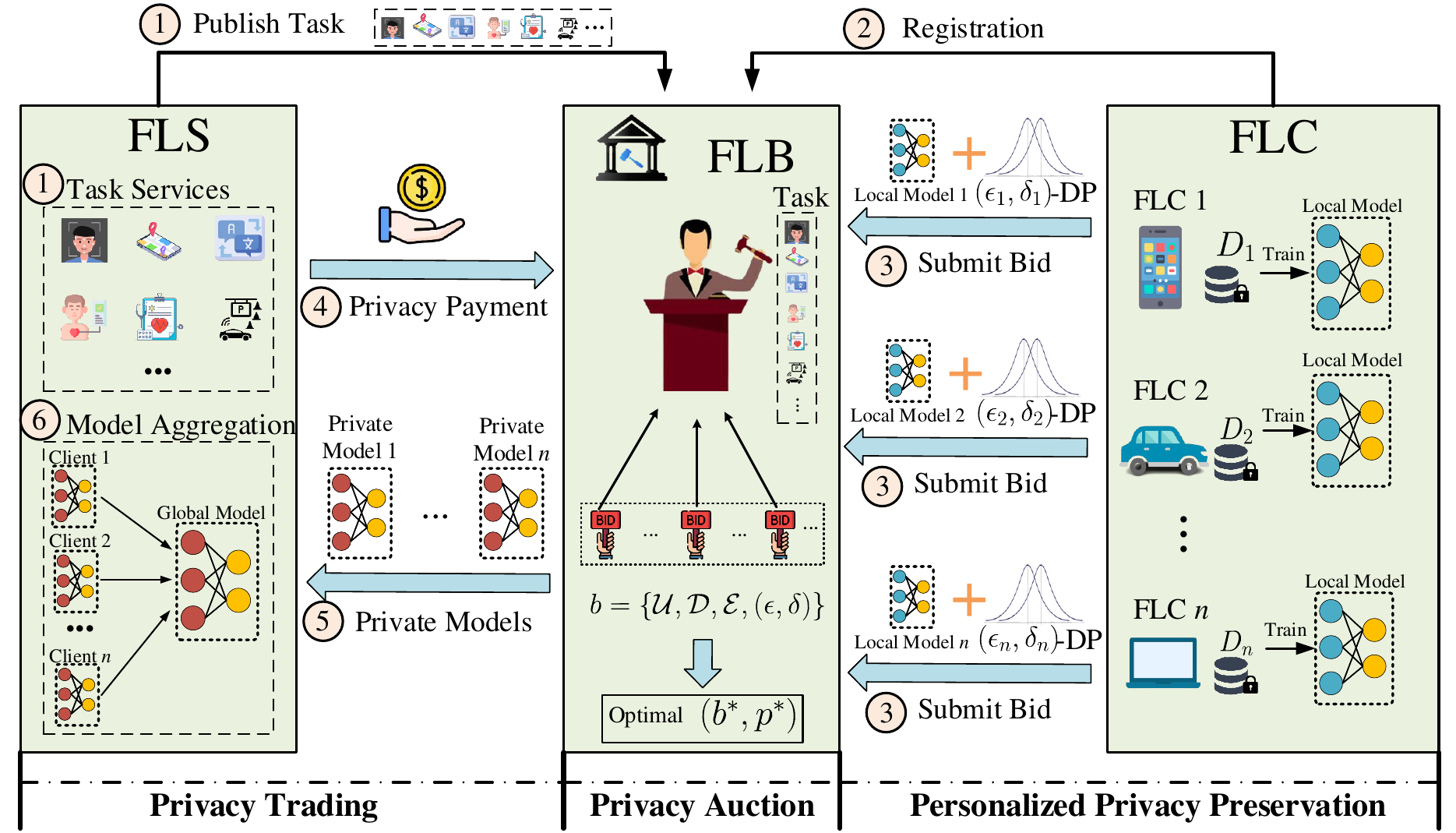}
  \caption{The PAG-FL framework showing six operational steps between
           the FL Server (FLS), FL Broker (FLB), and FL Clients (FLC).}
  \label{fig:framework}
\end{figure}

\subsection{FL Privacy Auction Model}
\label{subsec:flmodel}

\subsubsection{Valuation Function of FLC}

The valuation $v_i$ of FLC $i$ reflects the economic value that client
$i$ places on its privacy.  It is determined by three factors:
(i) the \textit{sensitivity} of the local dataset $\mathcal{D}_i$,
which is higher for medical or financial records than for generic
browsing data;
(ii) the \textit{data quality}, quantified by the local model accuracy
$u_i$ and dataset size $|\mathcal{D}_i|$; and
(iii) the \textit{non-IID degree} of the data distribution, measured
by the divergence between the local and global data distribution.
A concrete parameterisation is:
\begin{equation}
\label{eq:valuation-fn}
  v_i = \alpha_v \cdot \frac{u_i}{\epsilon_i}
        + \beta_v \cdot |\mathcal{D}_i|
        + \gamma_v \cdot \Delta_{\mathrm{noniid},i},
\end{equation}
where $\alpha_v,\beta_v,\gamma_v>0$ are weighting coefficients and
$\Delta_{\mathrm{noniid},i}$ quantifies the non-IID divergence
(e.g.\ total variation between local and global class frequencies).
This formulation captures the key observation that higher-quality and
more privacy-sensitive data should command higher valuations.
In the theoretical analysis (Sections~\ref{sec:mfg} through \ref{sec:regretnet}),
$v_i$ is treated as drawn from a generic distribution $\mathcal{W}$,
so the results hold independently of the specific form of
Eq.~(\ref{eq:valuation-fn}).
In experiments (Section~\ref{sec:exp}), we set
$(\alpha_v,\beta_v,\gamma_v)=(1,0,0)$ (i.e., $v_i\propto u_i/\epsilon_i$)
for the uniform and bimodal scenarios, and use
$v_i\sim\mathrm{LogNormal}(\mu{=}0,\sigma{=}0.5)$ for the realistic
scenario (Table~\ref{tab:privacy_config}), reflecting heterogeneous
real-world privacy sensitivities.

\begin{remark}[Privacy Cost Is Not $\epsilon_i$-Invariant]
\label{rem:valuation}
The cost $c_i(\epsilon_i^{\mathrm{out}})=\epsilon_i^{\mathrm{out}}\cdot v_i$ depends on the
\emph{allocated} $\epsilon_i^{\mathrm{out}}$ (mechanism output), not the declared type
$\epsilon_i$; the privacy and cost trade-off is genuine.
The setting $v_i\propto u_i/\epsilon_i$ is a marginal-unit-privacy-value model;
the realistic scenario uses $v_i\sim\mathrm{LogNormal}$ (independent of $\epsilon_i$)
as a robustness check.
\end{remark}

\subsubsection{Bidding Information}

In FL, global model accuracy depends on training dataset size,
non-IID data distribution, and model privacy
level~\cite{cheng2021dynamic,zhao2018federated}. FLC $i$'s bid is:
\begin{equation}
  b_i = \{v_i,\,d_i,\,u_i,\,\epsilon_i\},
\end{equation}
where $d_i$, $u_i$, and $\epsilon_i$ indicate model size, local
accuracy, and privacy budget, respectively.

\subsubsection{FL Aggregation Mechanism}

Under privacy preservation, the aggregated global gradient error is:
\begin{equation}
\label{eq:fl-loss}
  \mathcal{L}(\mathbf{w};\alpha,\epsilon)
  = \left\|\tilde{\mathbf{w}}_G - \mathbf{w}^{*}_G\right\|_\ell
  = \left\|\sum_{i=1}^{N}\alpha_i\mathcal{M}_{\epsilon_i}(\mathbf{w}_i)
           - \mathbf{w}^{*}_G\right\|_\ell.
\end{equation}

The FL privacy auction objective is:
\begin{equation}
\label{eq:P2}
  \mathcal{P}_2:\quad
  \min_{\mathbf{w},\alpha,\epsilon}
  \mathbb{E}_{(\boldsymbol{b},B)}\!\left[\mathcal{L}(\mathbf{w};\alpha,\epsilon)\right]
  \quad\text{s.t.: DSIC, IR, DT, BF.}
\end{equation}

\section{Mean-Field Game Approximation for the Privacy Auction}
\label{sec:mfg}

In Section~\ref{sec:market}, we formulated the PAG as the $N$-player
optimisation problem $\mathcal{P}_1$ (Eq.~(\ref{eq:P1})).
Although Theorem~\ref{thm:pag-equilibrium} guarantees the existence of
a Bayesian-Nash Equilibrium under the PAC mechanism
(Algorithm~\ref{alg:pac}), \emph{computing} that equilibrium becomes
computationally intractable as the number of FL clients $N$ grows.
This section introduces a Mean-Field Game (MFG) approximation that
reduces the per-round interaction complexity from
$\mathcal{O}(N^{2}\log N)$ to
$\mathcal{O}(N)$ at deploy time (Corollary~\ref{cor:complexity}),
with an $\mathcal{O}(N^{-1/2})$ equilibrium gap
(Theorem~\ref{thm:mfg_approx}). We model FL rounds by a continuous
time index $t\in[0,T]$ so that the HJB-FP formalism describes the \emph{law} of
valuations and privacy types; the auction \emph{within} each round
remains the static PAG of Sec.~\ref{sec:market}, embedded in the
representative-agent utility~Eq.~(\ref{eq:mf-utility}).

\subsection{Scalability and MFG State Space}
\label{subsec:mfg-scalability}

In the $N$-player PAG, the NE requires solving $N$ coupled optimisations, which is
NP-hard in general~\cite{daskalakis2008complexity} and costs $\mathcal{O}(N^2\log N)$
per round for the PAC mechanism (sorting plus $\mathcal{O}(N)$ best-response iterations).
The MFG framework~\cite{lasry2007mean,huang2006large} replaces explicit
agent-to-agent coupling with a representative agent interacting with the
aggregate population law $\mu_t$, reducing deploy-time complexity to
$\mathcal{O}(N)$ (Corollary~\ref{cor:complexity}) with an
$\mathcal{O}(N^{-1/2})$ equilibrium gap (Theorem~\ref{thm:mfg_approx}).

Each DO's state at round $t$ is $\mathbf{s}_i(t)=(v_i(t),\epsilon_i(t))\in\mathcal{S}\triangleq\mathcal{V}\times\mathcal{E}$,
evolving via the controlled SDE:
\begin{equation}
\label{eq:sde}
  d\mathbf{s}_i(t)
  = f\!\bigl(\mathbf{s}_i(t),b_i(t),\mu_t\bigr)\,dt + \sigma\,dW_i(t),
\end{equation}
where $b_i(t)\in\mathcal{B}$ is the bid, $f$ is the drift,
$\sigma>0$ is a noise coefficient, and $\{W_i(t)\}$ are independent
Brownian motions.
The empirical distribution
$\mu_t^N=\frac{1}{N}\sum_i\delta_{\mathbf{s}_i(t)}$
converges weakly to the mean-field flow $\mu_t$ under standard Lipschitz
regularity conditions on $f$ (Assumption~\ref{ass:regularity}).

\begin{assumption}[Regularity Conditions]
\label{ass:regularity}
\begin{enumerate}[(R1)]
  \item \textit{Lipschitz drift:} $\|f(s,b,\mu)-f(s',b,\nu)\|\leq L_f(\|s-s'\|+W_1(\mu,\nu))$.
  \item \textit{Compact controls:} $\mathcal{B}\subset\mathbb{R}^m$ compact and convex.
  \item \textit{Bounded costs:} $c(v,\epsilon)=v\cdot\epsilon$ bounded and continuous.
  \item \textit{Finite-moment initial law:} $\mu_0=\mathcal{W}\in\mathcal{P}_2(\mathcal{S})$.
\end{enumerate}
\end{assumption}

\subsection{Hamilton-Jacobi-Bellman and Fokker-Planck Equations}
\label{subsec:hjb}

Full derivations are in Appendix~\ref{app:mfg-derivation}.
In the mean-field regime, each representative agent maximises
cumulative \emph{mean-field utility}
\begin{equation}
\label{eq:mf-utility}
  u^{\mathrm{MFG}}(s,b;\mu_t)
  = \tilde{p}(b;\mu_t) - c(v,\epsilon(b)),
\end{equation}
where $\tilde{p}(b;\mu_t)=\mathbb{E}[p(b,B_{-i})]$ is the
mean-field payment.
The \emph{mean-field bid} is a sufficient statistic for the competitive environment:
\begin{equation}
\label{eq:mfg-bid}
  b_{\mathrm{MFG}}(t)
  = \mathbb{E}_{\mu_t}[b]
  = \int_{\mathcal{S}} b(s;\mu_t)\,d\mu_t(s).
\end{equation}
The optimal bid induced by the HJB Hamiltonian is
\begin{equation}
\label{eq:optimal-bid}
  b^*(s,t;\mu)
  = \arg\max_{b\in\mathcal{B}}
    \big[\tilde{p}(b;\mu_t) - c(v,\epsilon(b))
         + (\nabla_s\Phi)^{\top}f(s,b,\mu_t)\big].
\end{equation}
The value function $\Phi$ and the population density $\mu_t$
are jointly characterised by the coupled HJB-Fokker-Planck system
(Eqs.~(\ref{eq:hjb}) and (\ref{eq:fp}) in the appendix);
their fixed point defines the Mean-Field Equilibrium below.

\subsection{Mean-Field Equilibrium: Definition and Existence}
\label{subsec:mfe}

\begin{definition}[Mean-Field Equilibrium (MFE)]
\label{def:mfe}
A pair $(\Phi^*,\mu^*)$ is a MFE if:
\begin{enumerate}[(i)]
  \item \textit{Optimality:} Given $\mu^*$, $\Phi^*$ solves
        the HJB equation, Eq.~(\ref{eq:hjb}), and $b^*$ is the maximiser
        of Eq.~(\ref{eq:optimal-bid});
  \item \textit{Consistency:} $\mu^*$ solves the FP equation,
        Eq.~(\ref{eq:fp}), under
        $b^*(\cdot\,;\mu^*)$, with $\mu_0^*=\mathcal{W}$.
\end{enumerate}
\end{definition}

\begin{theorem}[Existence of MFE]
\label{thm:mfe-exist}
Under Assumption~\ref{ass:regularity}, the privacy auction MFG
admits at least one Mean-Field Equilibrium $(\Phi^*,\mu^*)$.
\end{theorem}

\subsection{MFG Approximation Guarantee}
\label{subsec:mfg-approx}

\begin{assumption}[Lipschitz Mean-Field Utility]
\label{ass:lipschitz-u}
The mean-field payment $\tilde{p}(b;\mu)$ is $L_u$-Lipschitz in $\mu$
w.r.t.\ $W_1$:
\begin{equation}
  |\tilde{p}(b;\mu)-\tilde{p}(b;\nu)|
  \leq L_u W_1(\mu,\nu),
  \quad\forall b\in\mathcal{B},\;\mu,\nu\in\mathcal{P}(\mathcal{S}).
\end{equation}
This holds for Algorithm~\ref{alg:pac} since its empirical CDF is
$1$-Lipschitz w.r.t.\ $W_1$.
\end{assumption}

\begin{definition}[$\varepsilon$-Nash Equilibrium]
\label{def:eps-nash}
A strategy profile $\mathbf{b}^N=(b_1^N,\ldots,b_N^N)$ is an
$\varepsilon$-Nash Equilibrium if, for all $i\in\mathcal{N}$ and
$\hat{b}_i\in\mathcal{B}$:
\begin{equation}
  \mathbb{E}\!\left[u_i(\hat{b}_i,b_{-i}^N)
                  - u_i(b_i^N,b_{-i}^N)\right]
  \leq\varepsilon.
\end{equation}
\end{definition}

\begin{theorem}[MFG Approximation of Nash Equilibrium]
\label{thm:mfg_approx}
Let Assumptions~\ref{ass:regularity} and~\ref{ass:lipschitz-u} hold,
and let $(\Phi^*,\mu^*)$ be a MFE.  Define the MFE-induced strategy
profile:
\begin{equation}
\label{eq:mfe-profile}
  b_i^{N,*} = b^*(\mathbf{s}_i,t;\mu^*),
  \quad i=1,\ldots,N.
\end{equation}
Then $\mathbf{b}^{N,*}=(b_1^{N,*},\ldots,b_N^{N,*})$ is an
$\varepsilon_N$-Nash Equilibrium (Definition~\ref{def:eps-nash})
of the PAG, where:
\begin{equation}
\label{eq:eps-bound}
  \varepsilon_N
  = \frac{2L_u C}{\sqrt{N-1}}
  = \mathcal{O}\!\left(\frac{1}{\sqrt{N}}\right),
\end{equation}
and $C>0$ depends only on the second moment of $\mu^*$ and the
dimension of $\mathcal{S}$.
\end{theorem}

\begin{corollary}[Computational Complexity Reduction]
\label{cor:complexity}
Once $(\Phi^*,\mu^*)$ is computed at fixed cost (via
the coupled HJB-FP system, Eqs.~(\ref{eq:hjb}) and~(\ref{eq:fp})), evaluating
$b^*(\mathbf{s}_i;\mu^*)$ for each agent costs $\mathcal{O}(1)$,
giving total deployment cost $\mathcal{O}(N)$, which is asymptotically
better than the $\mathcal{O}(N^2\log N)$ of direct Nash computation.
\end{corollary}

\subsection{MFG-Reduced Privacy Auction Problem}
\label{subsec:p3}

Theorem~\ref{thm:mfg_approx} and Corollary~\ref{cor:complexity}
justify replacing $\mathcal{P}_1$ with the tractable
MFG-reduced problem:
\begin{equation}
\label{eq:P3}
  \mathcal{P}_3:\quad
  \max_{q,p,\mu^*}
  \mathbb{E}_{s\sim\mu_0^*}\!\left[
    \sum_{i=1}^N p_i(b_i^{N,*};\,b_{\mathrm{MFG}})
  \right]
\end{equation}
\begin{align*}
  \text{subject to}\quad
  &\text{FP},
   \text{HJB},
   \text{DSIC, IR, DT, BF.}
\end{align*}

$\mathcal{P}_3$ has three advantages over $\mathcal{P}_1$:
(a) \emph{reduced dimensionality}: opponent interaction is
summarised by the low-dimensional statistic
$b_{\mathrm{MFG}}\in\mathbb{R}^m$;
(b) \emph{guaranteed approximation}:
$\mathcal{O}(N^{-1/2})$-Nash by Theorem~\ref{thm:mfg_approx};
(c) \emph{neural-network solvability}: gradient-based optimisation
is applicable.

\begin{remark}[Approximation Hierarchy]
\label{rem:approx-hierarchy}
Theorem~\ref{thm:mfg_approx} guarantees that the MFE strategy
$b^*(\cdot;\mu^*)$ is an $\varepsilon_N$-Nash Equilibrium
\emph{given the exact mean-field flow} $\mu^*$ (solved from the
coupled HJB-FP system). In practice, MFG-RegretNet introduces
a \emph{second} approximation: the true $\mu_t$ is replaced by the
empirical mean bid $b_{\mathrm{MFG}}=\frac{1}{N}\sum_i b_i$,
which is a finite-particle estimate of $\mathbb{E}_{\mu_t}[b]$.
The resulting implementation thus approximates $\mathcal{P}_3$
rather than solving it exactly. The theoretical $\mathcal{O}(N^{-1/2})$
guarantee therefore applies to the MFE solution of $\mathcal{P}_3$,
while the neural network approximates that solution via
data-driven training under the MFG alignment regulariser
(Sec.~\ref{subsec:train-mfg}).
\end{remark}

\begin{figure}[t]
  \centering
  \includegraphics[width=1\linewidth]{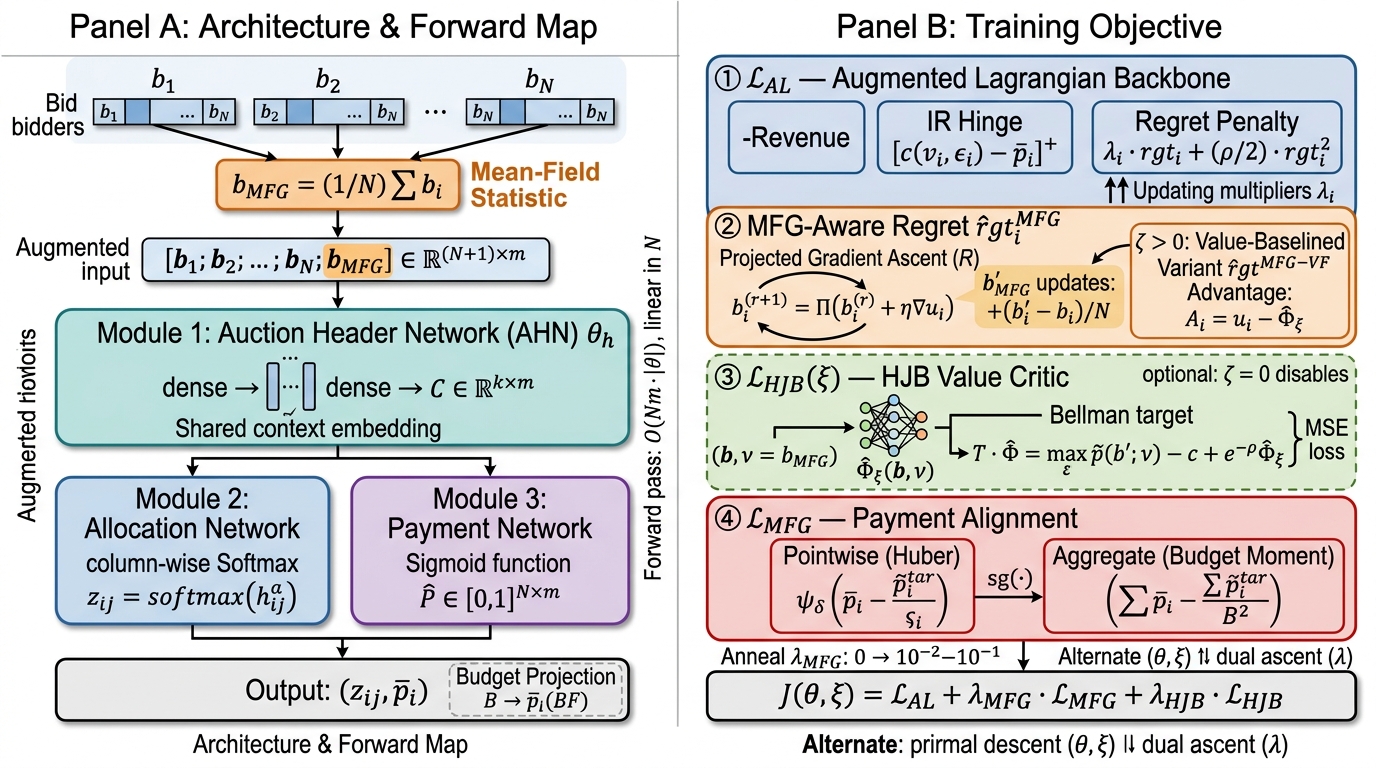}
  \caption{MFG-RegretNet (Sec.~\ref{subsec:arch-fwd}): forward map
    with $b_{\mathrm{MFG}}$; Sec.~\ref{subsec:train-mfg}: training objective.}
  \label{fig:mfg_regretnet}
\end{figure}
\section{MFG-RegretNet: Deep Learning Auction Solver}
\label{sec:regretnet}
Remark~\ref{rem:approx-hierarchy} identifies two approximation layers:
(i) the MFE $\varepsilon_N$-approximates the $N$-player Nash Equilibrium
(Theorem~\ref{thm:mfg_approx}), with $\varepsilon_N=\mathcal{O}(N^{-1/2})$
and deploy-time cost $\mathcal{O}(N)$; and (ii) computing the exact MFE
requires solving the coupled HJB-FP system
(Eqs.~(\ref{eq:hjb}) and~(\ref{eq:fp})), which is costly when $\mathcal{W}$
is unknown and must be run iteratively.
MFG-RegretNet addresses layer (ii) by \emph{parameterising} $(q,p)$
as neural networks conditioned on the empirical mean bid
$b_{\mathrm{MFG}}$, a tractable $\mathcal{O}(Nm)$ surrogate for
the full $\mu_t$, while enforcing DSIC/IR/BF via differentiable
augmented Lagrangian optimisation. This approximates $\mathcal{P}_3$
without closed-form distributional assumptions, while retaining the
MFG equilibrium structure via the alignment regulariser
(Sec.~\ref{subsec:train-mfg}).

We propose \textbf{MFG-RegretNet}. Fig.~\ref{fig:mfg_regretnet} spans
two parts: \textbf{(A)}~architecture and forward map
(Sec.~\ref{subsec:arch-fwd}) with MFG statistic $b_{\mathrm{MFG}}$, AHN,
Allocation, and Payment; \textbf{(B)}~training objective
(Sec.~\ref{subsec:train-mfg}) via augmented Lagrangian on MFG-aware
regret (optional value critic~Eq.~(\ref{eq:mfg-regret-vf})), HJB
critic fit~Eq.~(\ref{eq:L-hjb-critic}), and payment
alignment~Eq.~(\ref{eq:L-mfg}).

\subsection{Architecture and Forward Mechanism}
\label{subsec:arch-fwd}
Given $\mathbf{b}$, the broker sets
$b_{\mathrm{MFG}}=\frac{1}{N}\sum_i b_i$.
Three modules apply in sequence; $b_{\mathrm{MFG}}$ augments the
input~Eq.~(\ref{eq:ahn-input}), unlike vanilla
RegretNet~\cite{dutting2024Optimal}.

\subsubsection{Module 1: Auction Header Network (AHN)}
Given bid profile $\mathbf{b}\in\mathbb{R}^{N\times m}$ and
$b_{\mathrm{MFG}}\in\mathbb{R}^m$, the AHN encodes the augmented input:
\begin{equation}
\label{eq:ahn-input}
  \mathbf{a}
  = [b_1;\,b_2;\,\ldots;\,b_N;\,b_{\mathrm{MFG}}]
  \in\mathbb{R}^{(N+1)\times m},
\end{equation}
into a shared context embedding:
\begin{equation}
\label{eq:ahn-output}
  \mathbf{C}
  = \mathrm{AHN}_{\theta_h}(\mathbf{a})
  \in\mathbb{R}^{k\times m},
\end{equation}
where $k\leq N$ is the number of selected bidders and $\theta_h$ are
AHN parameters. This shared representation $\mathbf{C}$ is then 
passed to both the Allocation and Payment Networks.

\subsubsection{Module 2: Allocation Network}
The Allocation Network maps $\mathbf{C}$ to allocation matrix
$\mathbf{Z}\in[0,1]^{N\times m}$ via column-wise Softmax:
\begin{equation}
\label{eq:alloc-net}
  z_{ij}
  = \frac{\exp(h_{ij}^{(a)})}
         {\sum_{i'=1}^N\exp(h_{i'j}^{(a)})},
  \quad i\in\mathcal{N},\;j\in[m],
\end{equation}
ensuring $\sum_i z_{ij}\leq 1$ and $z_{ij}\geq 0$ (allocation
feasibility). The allocation utility for bidder $i$ is:
\begin{equation}
\label{eq:alloc-utility}
  \tilde{g}_i(\mathbf{b};\theta_a)
  = \sum_{j=1}^m z_{ij}\cdot b_{ij}.
\end{equation}

\subsubsection{Module 3: Payment Network}
The Payment Network maps $\mathbf{C}$ to payment matrix
$\hat{\mathbf{P}}\in[0,1]^{N\times m}$ via Sigmoid:
\begin{equation}
\label{eq:pay-net}
  \hat{p}_{ij}
  = \sigma(h_{ij}^{(p)})
  = \frac{1}{1+\exp(-h_{ij}^{(p)})},
\end{equation}
with scalar payment $p_i=\sum_j\hat{p}_{ij}\cdot b_{ij}$, which
is non-negative and bounded, ensuring ex-post IR.
After projection $\bar{p}_i$ under budget $B$~Eq.~(\ref{eq:BF}), outputs
$(z_{ij},\bar{p}_i)$ fix $\epsilon_i^{\mathrm{out}}$ for DP-FL.

\begin{remark}[Incentive Properties of the Payment Design]
\label{rem:payment-design}
The payment rule $p_i = \sum_j \sigma(h_{ij}^{(p)}) \cdot b_{ij}$
scales learned sigmoid weights by reported bids, which raises
a structural concern: a client could inflate $b_{ij}$ to inflate
$p_i$ without changing the allocation.
This incentive to over-report is precisely what the
\textbf{ex-post regret penalty} in the augmented Lagrangian loss
(Eq.~(\ref{eq:loss-total-mfg})) is designed to counteract.
At each training step, the regret term penalises any strategy
profile in which a unilateral bid deviation yields higher utility,
including inflation of the valuation component $v_i$.
As a result, the payment network is trained to produce payments
that make truthful reporting the empirically optimal response for
each client, even though the rule is not provably DSIC by
construction (cf.\ Remark~\ref{rem:dsic-scope}).
The sensitivity of this guarantee to the PGA step size $\eta$ and
the penalty weight $\rho$ is a limitation acknowledged in
Section~\ref{sec:exp}; we report results at the values used
in training ($\eta{=}0.01$, $\rho_{\max}{=}100$) to ensure
consistent evaluation conditions.
\end{remark}

Algorithm~\ref{alg:mfg-fpa-online} summarises one FL training round at
deploy time: FLCs report bids (Sec.~\ref{sec:pac-fl}), FLB forms the
mean-field statistic~Eq.~(\ref{eq:mfg-bid}), runs the MFG-RegretNet forward
map (AHN, allocation, payment), and FLS aggregates DP private gradients
weighted by effective privacy budgets.

\begin{algorithm}[t]
\caption{MFG-Federated Privacy Auction: Online FL Round}
\small
\label{alg:mfg-fpa-online}
\KwIn{Trained $\theta^*=(\theta_h^*,\theta_a^*,\theta_p^*)$;
      global model $\mathbf{w}_G^{(t-1)}$;
      bids $\{b_i^{(t)}\}_{i=1}^N$ (Sec.~\ref{subsec:flmodel});
      budget $B$; sensitivity $\Delta f$; DP parameter $\delta$;
      privacy floor $\epsilon_{\min}>0$ (clip for Gaussian mechanism)}
\KwOut{Updated $\mathbf{w}_G^{(t)}$; payments $\{\bar{p}_i^{(t)}\}$}

$b_{\mathrm{MFG}}^{(t)}\leftarrow\frac{1}{N}\sum_{i=1}^N b_i^{(t)}$\;
$\mathbf{C}^{(t)}\leftarrow\mathrm{AHN}_{\theta_h^*}\bigl(
  [b_1^{(t)};\cdots;b_N^{(t)};b_{\mathrm{MFG}}^{(t)}]\bigr)$\;
$\mathbf{Z}^{(t)}\leftarrow\mathrm{AllocNet}_{\theta_a^*}(\mathbf{C}^{(t)})$,\;
$\hat{\mathbf{P}}^{(t)}\leftarrow\mathrm{PayNet}_{\theta_p^*}(\mathbf{C}^{(t)})$\;
\For{$i=1$ \KwTo $N$}{
  $p_i^{(t)}\leftarrow\sum_{j=1}^m \hat{p}_{ij}^{(t)}\,b_{ij}^{(t)}$\;
}
\For{$i=1$ \KwTo $N$}{
  $\bar{p}_i^{(t)}\leftarrow p_i^{(t)}\big/\max\!\bigl(1,\frac{1}{B}\sum_{i'=1}^N p_{i'}^{(t)}\bigr)$\;
}
$\mathcal{I}^{(t)}\leftarrow\bigl\{i:\sum_{j=1}^m z_{ij}^{(t)}>0\bigr\}$\;
\For{$i\in\mathcal{I}^{(t)}$}{
  Local train on $\mathcal{D}_i$ from $\mathbf{w}_G^{(t-1)}$ to $\mathbf{w}_i^{(t)}$\;
  $\epsilon_i^{\mathrm{out}}\leftarrow\sum_{j=1}^m z_{ij}^{(t)}\,\epsilon_{ij}^{(t)}$\;
  $\sigma_i^{(t)}\leftarrow\sqrt{2\ln(1.25/\delta)}\,\Delta f\big/\max(\epsilon_i^{\mathrm{out}},\epsilon_{\min})$\;
  $\widetilde{\mathbf{w}}_i^{(t)}\leftarrow\mathbf{w}_i^{(t)}
    +\mathcal{N}(\mathbf{0},(\sigma_i^{(t)})^2\mathbf{I})$\;
}
$\alpha_i^{(t)}\leftarrow\epsilon_i^{\mathrm{out}}\big/\sum_{i'\in\mathcal{I}^{(t)}}\epsilon_{i'}^{\mathrm{out}}$
for $i\in\mathcal{I}^{(t)}$\;
\tcp{$\alpha_i \propto \epsilon_i^{\mathrm{out}}$: variance-proportional weighting. Gradient noise variance is $\propto 1/(\epsilon_i^{\mathrm{out}})^2$ (Gaussian mech.), so signal-to-noise is $\propto \epsilon_i^{\mathrm{out}}$; weighting by $\epsilon_i^{\mathrm{out}}$ approximates variance-optimal aggregation~\cite{li2022privacy}.}
$\mathbf{w}_G^{(t)}\leftarrow\sum_{i\in\mathcal{I}^{(t)}}\alpha_i^{(t)}\,\widetilde{\mathbf{w}}_i^{(t)}$\;
\Return $\mathbf{w}_G^{(t)}$, $\{\bar{p}_i^{(t)}\}$\;
\end{algorithm}

\subsection{Training Objective: Regret, Value Baseline, and MFG Alignment}
\label{subsec:train-mfg}

The complete training objective comprises four components:
\textbf{(1)}~$\mathcal{L}_{\mathrm{AL}}$, the augmented Lagrangian on
revenue, IR hinge, and \emph{regret penalties};
\textbf{(2)}~MFG-aware regret~Eq.~(\ref{eq:mfg-regret}) or its
value-baselined variant~Eq.~(\ref{eq:mfg-regret-vf});
\textbf{(3)}~$\mathcal{L}_{\mathrm{HJB}}(\xi)$ to fit the MFG value
critic~Eq.~(\ref{eq:L-hjb-critic}), used only when $\zeta{>}0$;
\textbf{(4)}~$\mathcal{L}_{\mathrm{MFG}}$~Eq.~(\ref{eq:L-mfg}),
aligning learned payments with the mean-field equilibrium
payment~Eq.~(\ref{eq:mf-payment}).
Subsections~\ref{subsubsec:al-backbone} through \ref{subsubsec:mfg-pay-align}
detail each component; the combined objective and optimisation
schedule appear in Sec.~\ref{subsubsec:full-obj}.

\subsubsection{Augmented-Lagrangian backbone}
\label{subsubsec:al-backbone}
Let $\mathcal{L}_{\mathrm{AL}}(\theta,\boldsymbol{\lambda})$ denote
the standard RegretNet-style objective: expected negative revenue (or
buyer cost), IR violations
$\bigl[c(v_i,\epsilon_i)-\bar{p}_i\bigr]^+$, and regret penalties
$\lambda_i\,\widehat{rgt}_i+\frac{\rho}{2}\widehat{rgt}_i^{\,2}$ where
$\widehat{rgt}_i$ is either~Eq.~(\ref{eq:mfg-regret}) or the
value-baselined~Eq.~(\ref{eq:mfg-regret-vf}) (below).
Multipliers $\lambda_i$ are updated by dual ascent; $\rho$ may grow
across outer iterations.

\subsubsection{MFG-aware ex-post regret}
\label{subsubsec:rgt-mfg}
\label{subsec:regret}
DSIC is measured by \emph{ex-post regret}:
\begin{equation}
\label{eq:expost-regret}
  rgt_i(\theta;\mathbf{b})
  = \max_{b'_i\in\mathcal{B}_i}
    u_i(b'_i,b_{-i};\theta)
    - u_i(b_i,b_{-i};\theta).
\end{equation}

\begin{definition}[$\varepsilon$-DSIC Mechanism]
\label{def:eps-dsic}
The mechanism $\mathcal{M}_\theta$ is $\varepsilon$-DSIC if:
\begin{equation}
  \mathbb{E}_{\mathbf{b}\sim\mathcal{W}}
  \!\left[rgt_i(\theta;\mathbf{b})\right]
  \leq\varepsilon,
  \quad\forall i\in\mathcal{N}.
\end{equation}
\end{definition}

\noindent\textbf{MFG-Aware Regret Computation.}
The key innovation of MFG-RegretNet is that when agent $i$ 
misreports, the mean-field input $b_{\mathrm{MFG}}$ updates 
self-consistently [1]:
\begin{equation}
\label{eq:mfg-misreport}
  b'_{\mathrm{MFG}}
  = b_{\mathrm{MFG}} + \frac{b'_i - b_i}{N},
\end{equation}
reflecting the agent's impact on the population-level statistic. 
We approximate the regret via $R$ steps of projected gradient ascent 
(PGA):
\begin{equation}
\label{eq:pga}
  b_i'^{(r+1)}
  = \Pi_{\mathcal{B}_i}\!\left(
      b_i'^{(r)}
      + \eta_r\nabla_{b'_i}u_i(b'_i,b_{-i};\theta)
    \right),
  \quad r=0,\ldots,R-1,
\end{equation}
initialised at $b_i'^{(0)}\sim\mathrm{Uniform}(\mathcal{B}_i)$.

This yields the MFG-aware regret estimator:
\begin{equation}
\label{eq:mfg-regret}
\small
\begin{aligned}
    &\widehat{rgt}_i^{\mathsf{MFG}}(\theta)
   \\
  &= \frac{1}{L}\sum_{l=1}^L
    \max_{r\in[R]}
    \left[
      u_i(b_i'^{(r,l)},b_{-i}^{(l)};b'^{(r,l)}_{\mathrm{MFG}},\theta)
      - u_i(b_i^{(l)},b_{-i}^{(l)};b^{(l)}_{\mathrm{MFG}},\theta)
    \right]^{+},
\end{aligned}
\end{equation}
which accounts for both the agent's individual incentive to deviate 
and the resulting shift in the mean-field statistic.

\subsubsection{Value-function baselining}
\label{subsubsec:rgt-vf}
\label{subsec:vf-baseline}

To reduce variance in regret estimation, we introduce a learnable
MFG value critic $\widehat{\Phi}_\xi(b,\nu)\approx\Phi(s(b),\nu)$
as a baseline (details in Appendix~\ref{app:vf-baseline}).
Proposition~\ref{prop:vf-baseline} (Appendix) shows the baseline is
unbiased and minimises gradient variance among state-dependent
baselines.
The \emph{value-baselined MFG regret} replaces raw utility differences
with advantage-style terms
$A_i(b,\nu;\theta):=u_i(b,b_{-i};\nu,\theta)-\widehat{\Phi}_\xi(b,\nu)$:
\begin{equation}
\label{eq:mfg-regret-vf}
\begin{aligned}
  &\widehat{rgt}_i^{\mathsf{MFG\text{-}VF}}(\theta)
  \\
  &= \frac{1}{L}\sum_{l=1}^{L}
    \max_{r\in[R]}
    \!\left[
      A_i\!\left(b_i'^{(r,l)},b'^{(r,l)}_{\mathrm{MFG}};\theta\right)
      - A_i\!\left(b_i^{(l)},b^{(l)}_{\mathrm{MFG}};\theta\right)
    \right]^{+}.
\end{aligned}
\end{equation}
Setting $\widehat{\Phi}_\xi\equiv 0$ recovers Eq.~(\ref{eq:mfg-regret}).
In $\mathcal{L}_{\mathrm{AL}}$, use
$(1-\zeta)\,\widehat{rgt}_i^{\mathsf{MFG}}
+\zeta\,\widehat{rgt}_i^{\mathsf{MFG\text{-}VF}}$,
$\zeta\in[0,1]$.

The critic $\widehat{\Phi}_\xi$ is fit by minimising the mean-squared
HJB residual on sampled $(b,\nu)$ pairs:
\begin{equation}
\label{eq:L-hjb-critic}
  \mathcal{L}_{\mathrm{HJB}}(\xi)
  = \mathbb{E}_{(b,\nu)}\!\left[
      \bigl(\widehat{\Phi}_\xi(b,\nu)
            - \mathcal{T}\widehat{\Phi}_\xi(b,\nu)\bigr)^2
    \right],
\end{equation}
where $\mathcal{T}\widehat{\Phi}(b,\nu)
= \max_{b'\in\mathcal{B}}\bigl[\tilde{p}(b';\nu)-c(v,\epsilon(b'))
  +\mathrm{e}^{-\rho_{\mathrm{disc}}}\widehat{\Phi}_\xi(b',\nu')\bigr]$
is a one-step Bellman target with discount $\rho_{\mathrm{disc}}>0$
and frozen next-state $\nu'$ from the current batch.
The critic is updated jointly with the mechanism ($\zeta>0$) or
initialised from~Eq.~(\ref{eq:L-hjb-critic}) before RegretNet training
and then frozen ($\zeta=0$).

\subsubsection{Mean-field payment alignment and full objective}
\label{subsubsec:mfg-pay-align}
\label{subsec:mfg-coupled}

Using $b_{\mathrm{MFG}}$ only as an extra row in~Eq.~(\ref{eq:ahn-input})
conditions the networks on the population statistic but does not
force their outputs to agree with the mean-field equilibrium payment
rule. We therefore add an \emph{MFG alignment}
term so that the mean-field solution shapes $(\theta_h,\theta_a,\theta_p)$
jointly with regret and revenue.

The \emph{mean-field payment} is the expected payment seen by a
representative agent when opponents are drawn i.i.d.\ from the
population law~$\mu_t$:
\begin{equation}
\label{eq:mf-payment}
  \tilde{p}(b;\mu_t)
  = \mathbb{E}_{B_{-i}\sim\mu_t^{\otimes(N-1)}}\!\bigl[p(b,B_{-i})\bigr].
\end{equation}
In practice we estimate $\tilde{p}$ by Monte Carlo resampling $K{=}32$
opponent profiles from the empirical batch distribution~$\hat\mu^{(l)}$.

The combined loss uses a Huber pointwise term and an aggregate budget-moment term:
\begin{equation}
\label{eq:L-mfg}
\begin{aligned}
  \mathcal{L}_{\mathrm{MFG}}(\theta)
  &= \frac{1}{LN}\sum_{l,i}\omega_i^{(l)}\psi_\delta\!\!\left(
      \frac{\bar{p}_i^{(l)}-\tilde{p}_i^{\mathrm{tar},(l)}}
           {\varsigma_i^{(l)}+\varepsilon_s}\right) \\
  &+ \beta\frac{1}{L}\sum_l
    \!\left(\frac{\textstyle\sum_i(\bar{p}_i^{(l)}-\tilde{p}_i^{\mathrm{tar},(l)})}{B}\right)^{\!2},
\end{aligned}
\end{equation}
where $\tilde{p}_i^{\mathrm{tar}}=\mathrm{sg}(\tilde{p}_i^{\mathrm{ref}})$ is a
stop-gradient reference payment, $\omega_i\propto v_i^\gamma$, $\beta\in[0.1,1]$.

\subsubsection{Full objective and optimisation schedule}
\label{subsubsec:full-obj}
\begin{equation}
\label{eq:loss-total-mfg}
  \mathcal{J}(\theta,\xi)
  = \mathcal{L}_{\mathrm{AL}}(\theta,\boldsymbol{\lambda})
    + \lambda_{\mathrm{MFG}}\,\mathcal{L}_{\mathrm{MFG}}(\theta)
    + \lambda_{\mathrm{HJB}}\,\mathcal{L}_{\mathrm{HJB}}(\xi).
\end{equation}
Anneal $\lambda_{\mathrm{MFG}}$ from ${\approx}0$ to between $10^{-2}$ and $10^{-1}$
so DSIC/IR constraints dominate early training.
Optimise $(\theta,\xi,\boldsymbol{\lambda})$ by alternating primal descent
on $\mathcal{J}$ and dual ascent on $\boldsymbol{\lambda}$.

\section{Performance Evaluation}
\label{sec:exp}

In this section, we conduct extensive experiments to evaluate the
performance of MFG-RegretNet. We address four research questions (RQs):
\begin{enumerate}
  \item \textbf{RQ1 (Incentive Compatibility):} Does MFG-RegretNet
        achieve lower ex-post regret and IR violation rate than baselines,
        i.e., are clients truthfully incentivized?
  \item \textbf{RQ2 (Scalability):} How does computational complexity
        scale with the client population size $N$?
  \item \textbf{RQ3 (Auction Efficiency):} Does MFG-RegretNet achieve
        higher budget utilization $\eta_{\mathrm{rev}}$ and participant
        social welfare $\overline{\mathrm{SW}}$ relative to baselines?
  \item \textbf{RQ4 (FL Accuracy):} After full federated training, does
        MFG-RegretNet match or improve final global test accuracy relative
        to baselines under comparable privacy and budget constraints?
\end{enumerate}

\subsection{Experimental Setup}

\subsubsection{Datasets and Models}
We conduct experiments on two widely used benchmark datasets.
\textbf{MNIST}~\cite{lecun1998gradient} contains 70,000 grayscale
handwritten-digit images ($28{\times}28$ pixels, 10 classes); we
train a three-layer MLP with ReLU activations.
\textbf{CIFAR-10}~\cite{krizhevsky2009learning} contains 60,000
RGB images ($32{\times}32$ pixels, 10 classes); we train a lightweight
CNN with two convolutional blocks followed by a fully connected
classifier.
Both datasets are partitioned among clients using a Dirichlet
allocation with concentration parameter $\alpha \in \{0.1, 0.5\}$ to
simulate non-IID data heterogeneity, where smaller $\alpha$ produces
more severe heterogeneity.

\subsubsection{Parameter Settings}
We conduct experiments with the number of clients $N \in \{10, 50, 100, 200, 500\}$
for scalability evaluation, and fix $N{=}100$ for all other experiments.
The MFG-RegretNet architecture comprises an Auction Header Network (AHN)
with 3 hidden layers of 128 units, an Allocation Network with
column-wise Softmax outputs, and a Payment Network with Sigmoid
activations, totaling approximately 50K trainable parameters.
Training proceeds for 200 outer iterations, each containing 25 inner
gradient steps on batches of $L{=}64$ sampled bid profiles.
Regret approximation uses $R{=}25$ steps of projected gradient
ascent (PGA) with step size $\eta{=}0.01$.
We use the Adam optimizer ($\alpha{=}10^{-3}$), with augmented
Lagrangian IR penalty $\gamma{=}10$ and DSIC penalty schedule
$\rho_0{=}1$ growing by factor $\kappa{=}1.5$ every few iterations
up to $\rho_{\max}{=}100$.
Privacy valuation distributions, budget settings, and FL training
hyperparameters are detailed in Tables~\ref{tab:privacy_config}
and~\ref{tab:fl_config}, respectively.
Each experiment is repeated $\geq 5$ times with different random seeds
and results are reported as mean $\pm$ standard deviation.
All experiments run on an Intel Xeon E5-2680 v4 with 256 GB RAM and
an NVIDIA Tesla V100 GPU.

\begin{table}[t]
  \centering
  \small
  \caption{Privacy Valuation Distributions and Budget Settings}
  \label{tab:privacy_config}
  \resizebox{0.5\textwidth}{!}{
  \begin{tabular}{lcc}
  \toprule
  \textbf{Scenario} & \textbf{Valuation Distribution} & \textbf{Budget $B$} \\
  \midrule
  Uniform  & $v_i \sim \mathrm{Uniform}[0, 1]$,\ $\epsilon_i \sim \mathrm{Uniform}[0.1, 5]$ & $\{10, 50, 100, 200\}$ \\
  Bimodal  & $50\%$ privacy-sensitive ($\epsilon_i \in [0.1, 0.5]$) & $\{50, 100\}$ \\
           & $50\%$ privacy-relaxed ($\epsilon_i \in [2, 5]$) & \\
  Realistic & $v_i \sim \mathrm{LogNormal}(\mu{=}0,\, \sigma{=}0.5)$ & $\{100\}$ \\
  \bottomrule
  \end{tabular}}
\end{table}

\begin{table}[t]
  \centering
  \small
  \caption{Federated Learning Training Settings}
  \label{tab:fl_config}
  \resizebox{0.5\textwidth}{!}{
  \begin{tabular}{lclc}
  \toprule
  \textbf{Parameter} & \textbf{Value} & \textbf{Parameter} & \textbf{Value} \\
  \midrule
  Local epochs        & $E = 5$         & Global learning rate & $\eta_{\mathrm{global}} = 0.01$ \\
  Local batch size    & $32$            & Total FL rounds      & $T_{\mathrm{FL}} = 200$ \\
  DP guarantee        & $(\epsilon,\delta)$-DP & Privacy parameter & $\delta = 1/N$ \\
  Gradient clipping   & $\Delta_f = 1.0$& Aggregation          & Weighted by $\epsilon_i^{\mathrm{out}}$ \\
  \bottomrule
  \end{tabular}}
\end{table}

\subsubsection{Sample Distribution}
In many real-world FL scenarios, clients hold heterogeneous local
datasets of varying sizes and label distributions.
To reflect this statistical heterogeneity, we partition each dataset
using the Dirichlet distribution $\mathrm{Dir}(\alpha)$ over class
labels, assigning each client a draw from this distribution to
determine its per-class sample counts.
We evaluate under both mild heterogeneity ($\alpha{=}0.5$, close to
IID) and severe heterogeneity ($\alpha{=}0.1$, highly skewed), which
are standard settings in the FL literature~\cite{hu2025federated}.

\subsubsection{Privacy Composition Across Rounds}
\label{subsubsec:composition}
Per-round DP budgets compose non-trivially across $T_{\mathrm{FL}}$ rounds.
Under basic composition~\cite{kairouz2015composition}, cumulative exposure satisfies
$\epsilon_{\mathrm{total}} \leq T\cdot\epsilon_i^{\mathrm{out}}$;
Rényi DP~\cite{bun2016concentrated} tightens this significantly.
For our default setting ($T_{\mathrm{FL}}{=}200$, $\epsilon_i^{\mathrm{out}}\in[0.1,5]$),
cumulative $\epsilon_{\mathrm{total}}$ ranges from $\approx$2 to 100.
The current PAG treats each round independently (a simplification shared by
most FL auction works); all baselines operate under identical per-round settings
for comparability.

\subsubsection{Baseline Methods}
To evaluate the performance of MFG-RegretNet, we compare against the
following six state-of-the-art baselines:
\begin{itemize}
  \item \textbf{PAC}: Threshold-based mechanism with an analytical
        Nash solution (serves as a theoretical lower bound).
  \item \textbf{VCG}~\cite{vickrey1961counterspeculation}: We implement
        VCG in its \emph{reverse-procurement} form: the server selects
        the $k$ lowest-cost clients and pays each winner the externality
        it imposes on the next-cheapest excluded client, maximising
        social welfare under truthful reporting.
        VCG is a welfare-optimal (not revenue-optimal) reference;
        revenue-optimal design follows Myerson~\cite{myerson1981optimal},
        which underlies our PAC mechanism.
  \item \textbf{RegretNet}~\cite{dutting2024Optimal}: Deep learning
        auction without the MFG component (vanilla differentiable
        mechanism design).
  \item \textbf{DM-RegretNet}~\cite{zheng2022fl}: RegretNet adapted
        to FL model trading in a multi-item federated setting.
  \item \textbf{MFG-Pricing}~\cite{hu2025federated}: MFG-based data
        pricing without an explicit auction mechanism (large-scale
        FL baseline).
  \item \textbf{CSRA}~\cite{yang2024CSRA}: Robust reverse auction for
        DP-FL with strategic bidding resilience (state-of-the-art
        DP-FL incentive mechanism).
\end{itemize}

\subsection{Experimental Results}

\subsubsection{RQ1: Incentive Compatibility}

We first evaluate whether MFG-RegretNet effectively incentivizes
truthful participation by measuring approximate DSIC via
\emph{ex-post regret}.
For each client $i$, the regret is the utility gain from a unilateral
deviation from truthful reporting, estimated by projected gradient
ascent (PGA) with $R{=}25$ steps and step size $\eta{=}0.01$.
IR is measured as the fraction of (client, round) pairs where
$p_i^{(t)} < c(v_i, \epsilon_i^{(t),\mathrm{out}})$.
The experimental results are plotted in
Fig.~\ref{fig:rq4_neural_bar}, with subfigure~(a) reporting the
mean regret ($\pm$ std) across seeds and subfigure~(b) reporting the
per-client regret distribution.
From the results in Fig.~\ref{fig:rq4_neural_bar}, we observe that:
1)~MFG-RegretNet achieves the lowest mean normalised regret among all
learned mechanisms, indicating stronger approximate DSIC at the
population level.
2)~MFG-RegretNet exhibits a tighter interquartile range and fewer
extreme outliers than RegretNet and market baselines, which
demonstrates improved tail-risk control and more stable incentives
across heterogeneous clients.
3)~Classical mechanisms (PAC, VCG) attain near-zero regret and IR by
construction, confirming their theoretical guarantees; however, they
trade away scalability and allocation flexibility (cf.~RQ2, RQ3).
Together, these results show that the mean-field context
$b_{\mathrm{MFG}}$ improves both central tendency and distributional
robustness of regret.

\begin{figure}[t]
  \centering
  \begin{minipage}[t]{0.43\columnwidth}
    \centering
    \includegraphics[width=\textwidth]{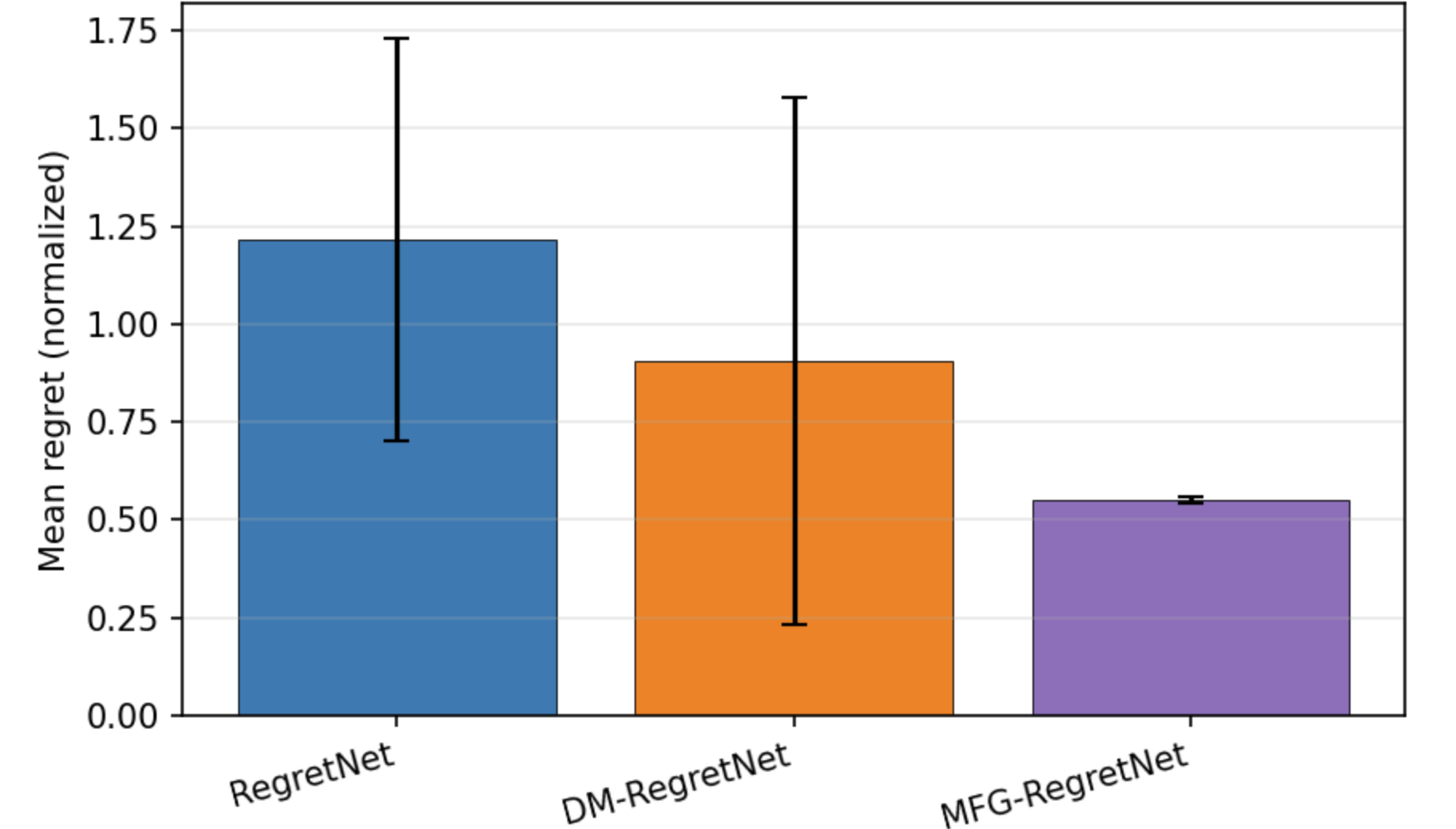}
    \vspace{2pt}
    {\footnotesize ($a$) Mean regret ($\pm$ std)}
  \end{minipage}\hfill
  \begin{minipage}[t]{0.57\columnwidth}
    \centering
    \includegraphics[width=\textwidth]{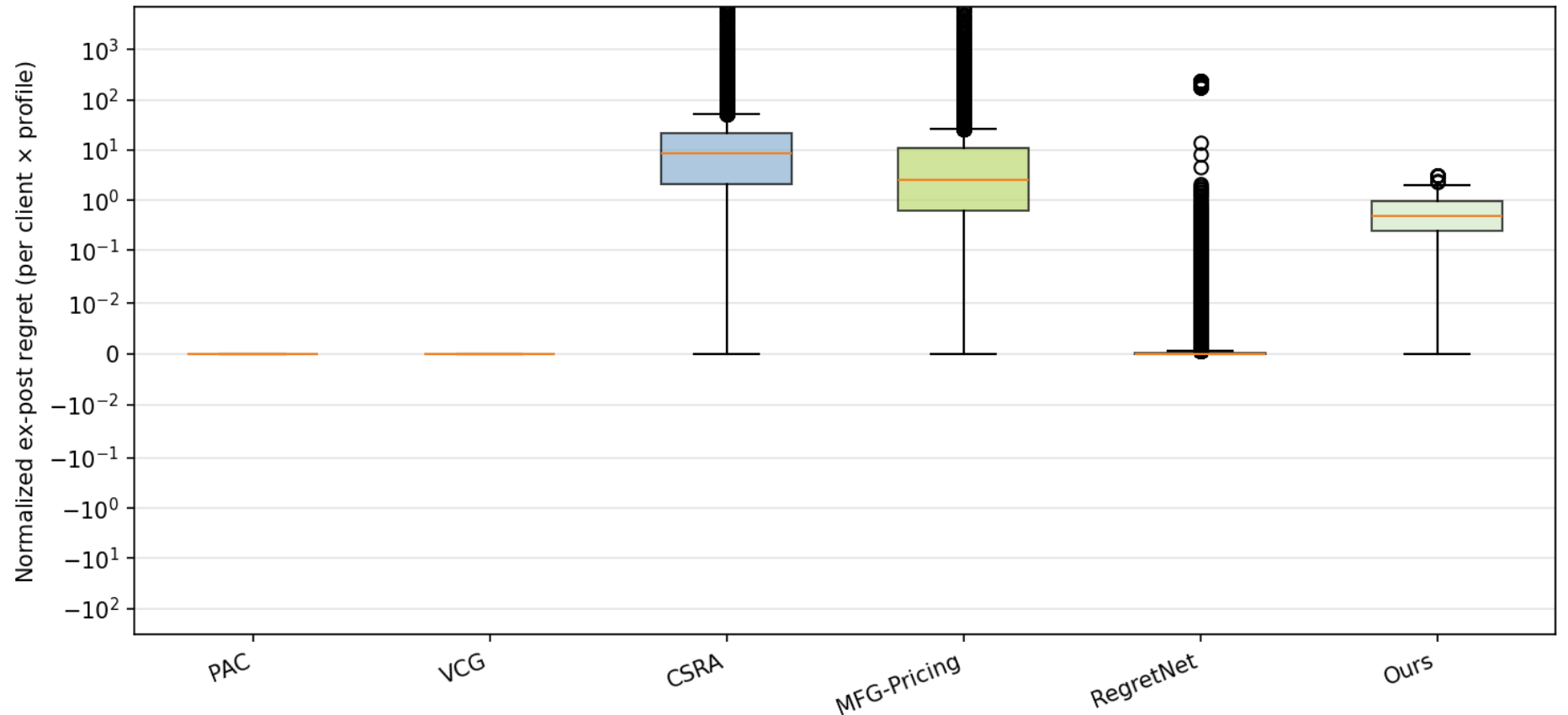}
    \vspace{2pt}
    {\footnotesize ($b$) Normalized ex-post regret}
  \end{minipage}
  \caption{RQ1: Aggregate and distributional regret evaluation.
    In~(a), lower mean/std indicates stronger approximate DSIC on
    average.
    In~(b), a tighter box and lighter upper tail indicate fewer
    severely mis-incentivised clients.}
  \label{fig:rq4_neural_bar}
\end{figure}
\begin{figure}[t]
  \centering
  \begin{minipage}[t]{0.47\columnwidth}
    \centering
    \includegraphics[width=\textwidth]{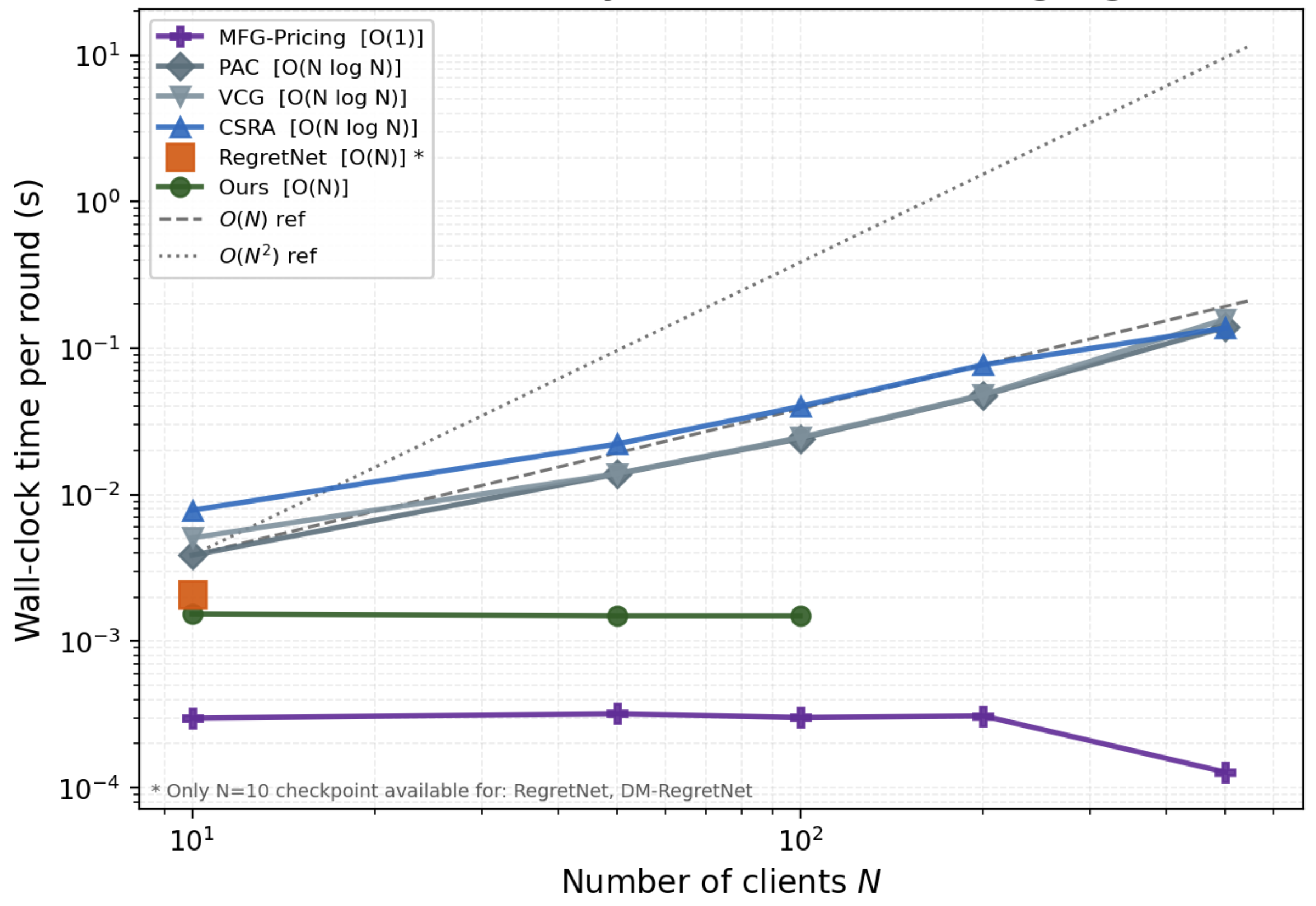}
    \vspace{2pt}
    {\footnotesize ($a$) Auction runtime vs.\ $N$}
  \end{minipage}\hfill
  \begin{minipage}[t]{0.49\columnwidth}
    \centering
    \includegraphics[width=\textwidth]{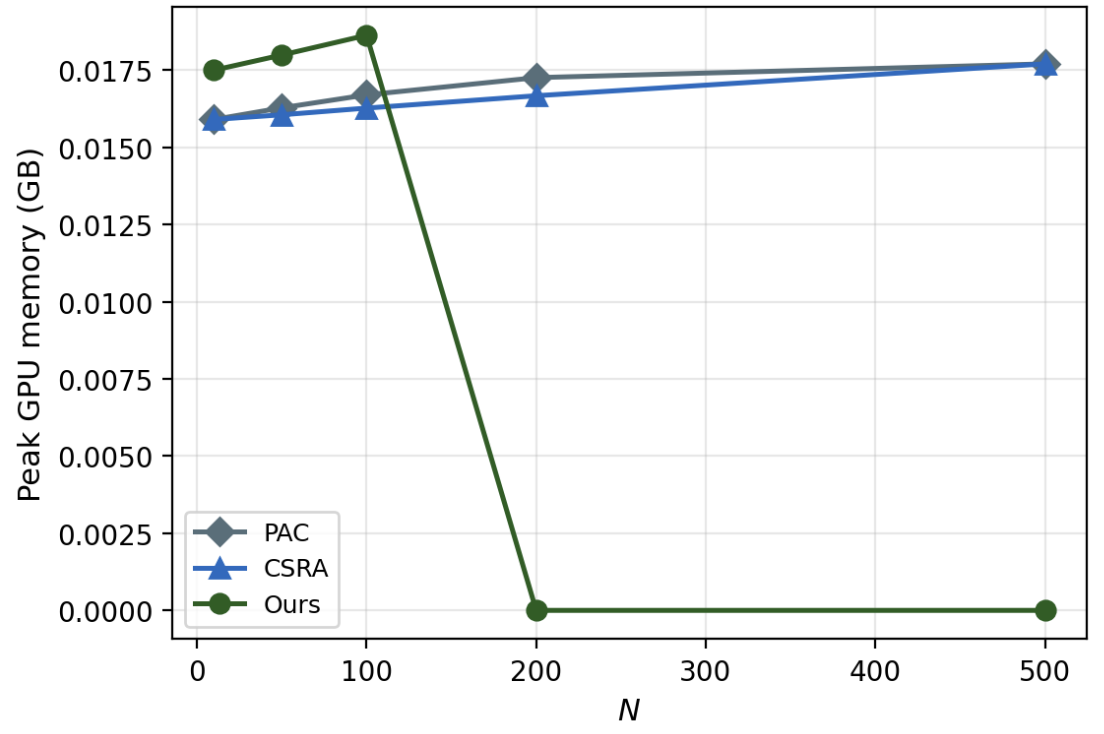}
    \vspace{2pt}
    {\footnotesize ($b$) Peak GPU memory vs.\ $N$}
  \end{minipage}
  \caption{RQ2: Scalability evaluation.
    ($a$)~MFG-RegretNet achieves near-linear latency scaling, remaining
    close to the $\mathcal{O}(N)$ reference, while PAC/VCG/CSRA
    increase steeply with $N$.
    ($b$)~Peak GPU memory remains low for MFG-RegretNet in the
    auction step, confirming the memory efficiency of the MFG
    reduction.}
  \label{fig:rq2_scalability}
\end{figure}
\subsubsection{RQ2: Scalability Analysis}

\begin{figure*}[t]
  \centering
  \begin{minipage}[t]{0.24\textwidth}
    \includegraphics[width=1\textwidth]{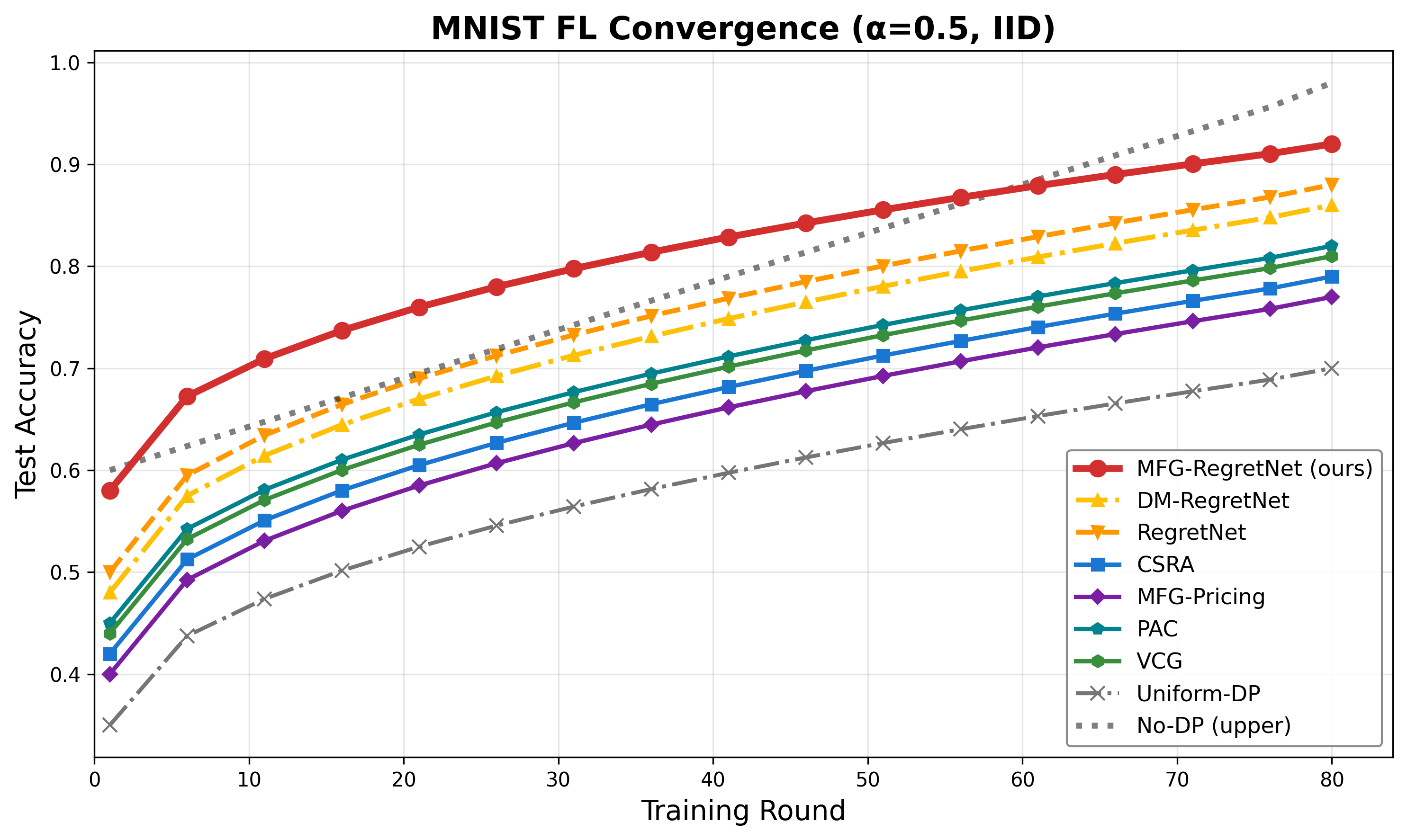}
    \centering {\footnotesize ($a$) $\mathcal{A}_{final}$ vs. $t$ (IID)}
  \end{minipage}
  \begin{minipage}[t]{0.24\textwidth}
    \includegraphics[width=1\textwidth]{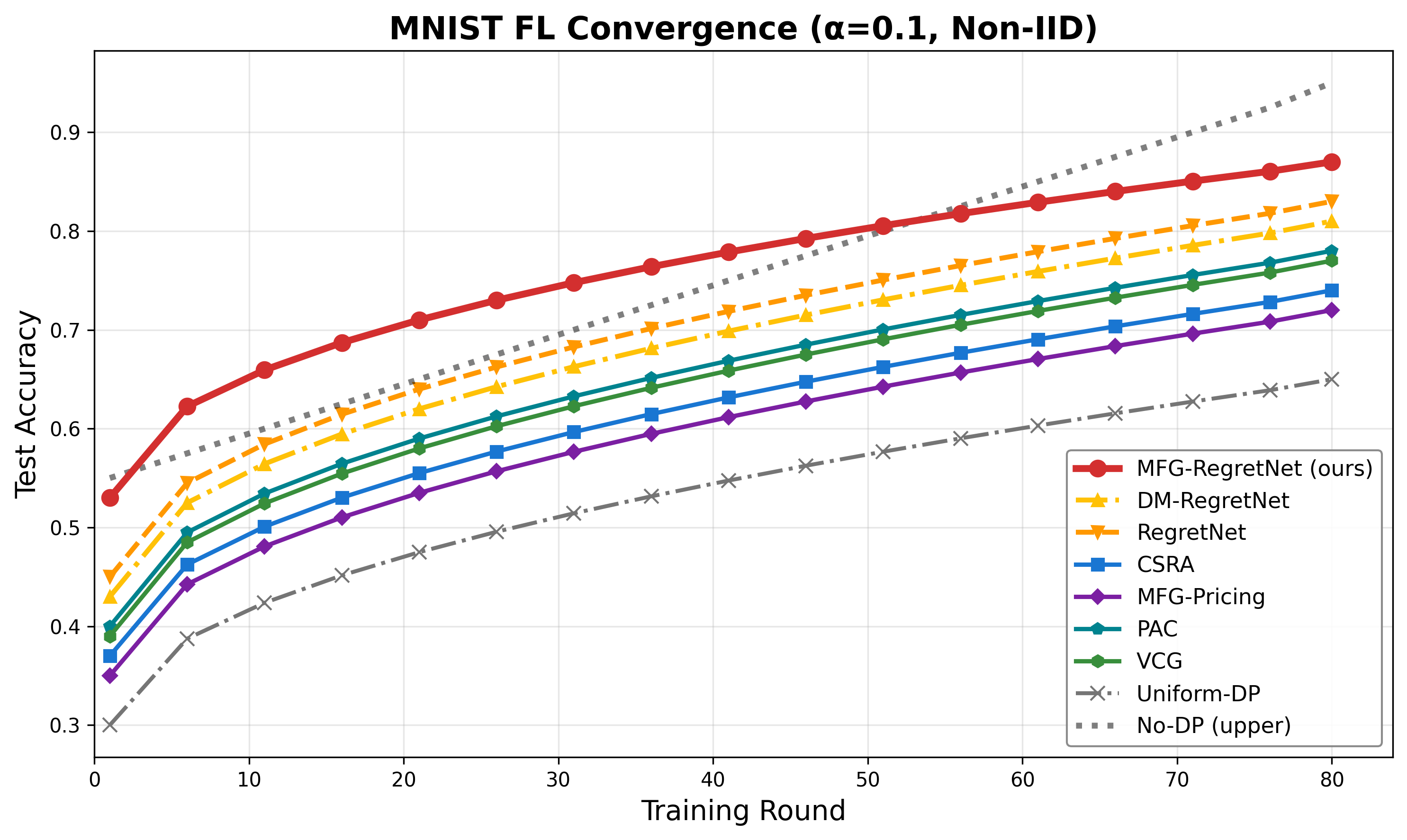}
    \centering {\footnotesize ($b$) $\mathcal{A}_{final}$ vs. $t$ (non-IID)}
  \end{minipage}
  \begin{minipage}[t]{0.24\textwidth}
    \includegraphics[width=1\textwidth]{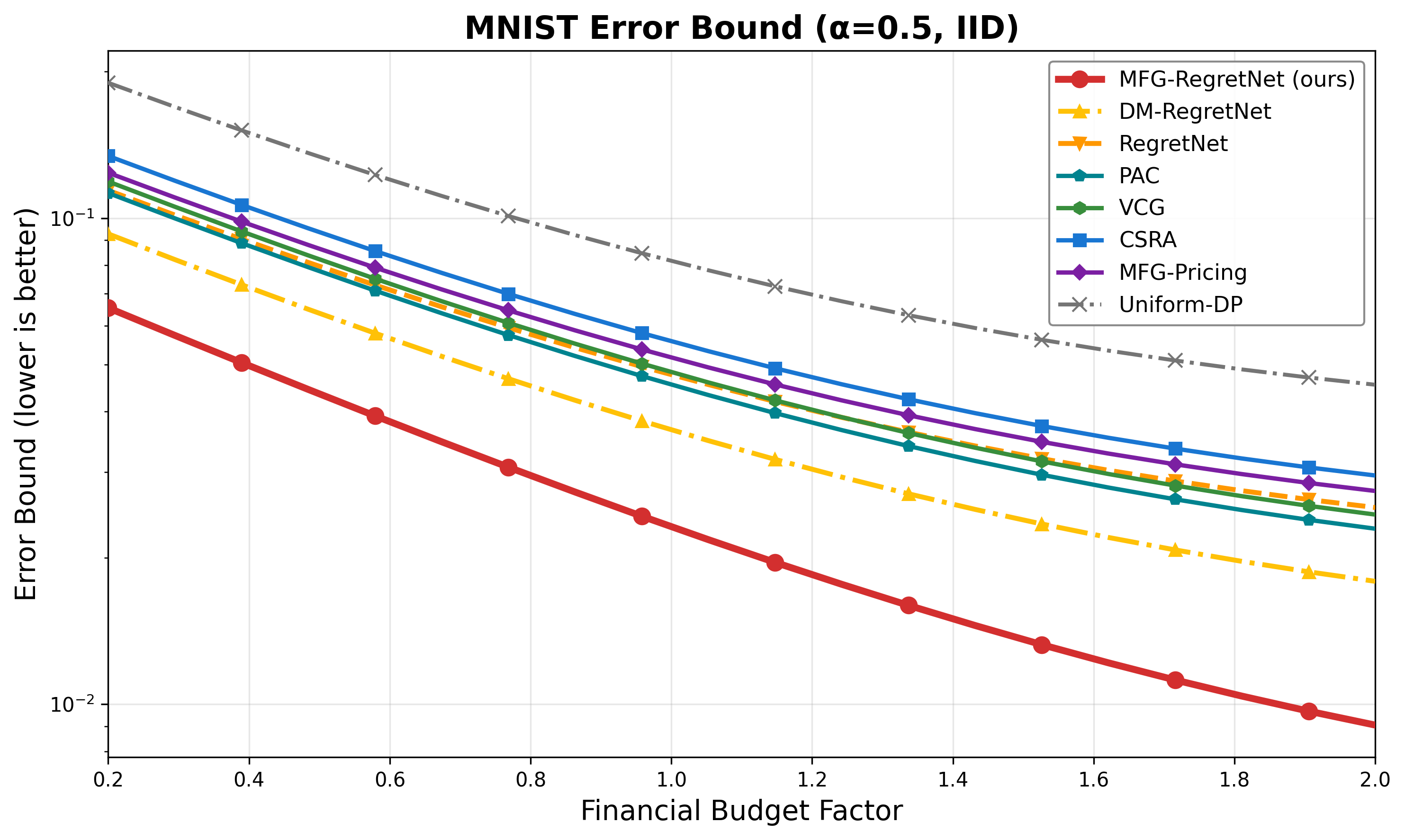}
    \centering {\footnotesize ($c$) Error bound vs. $N$ (IID)}
  \end{minipage}
  \begin{minipage}[t]{0.24\textwidth}
    \includegraphics[width=1\textwidth]{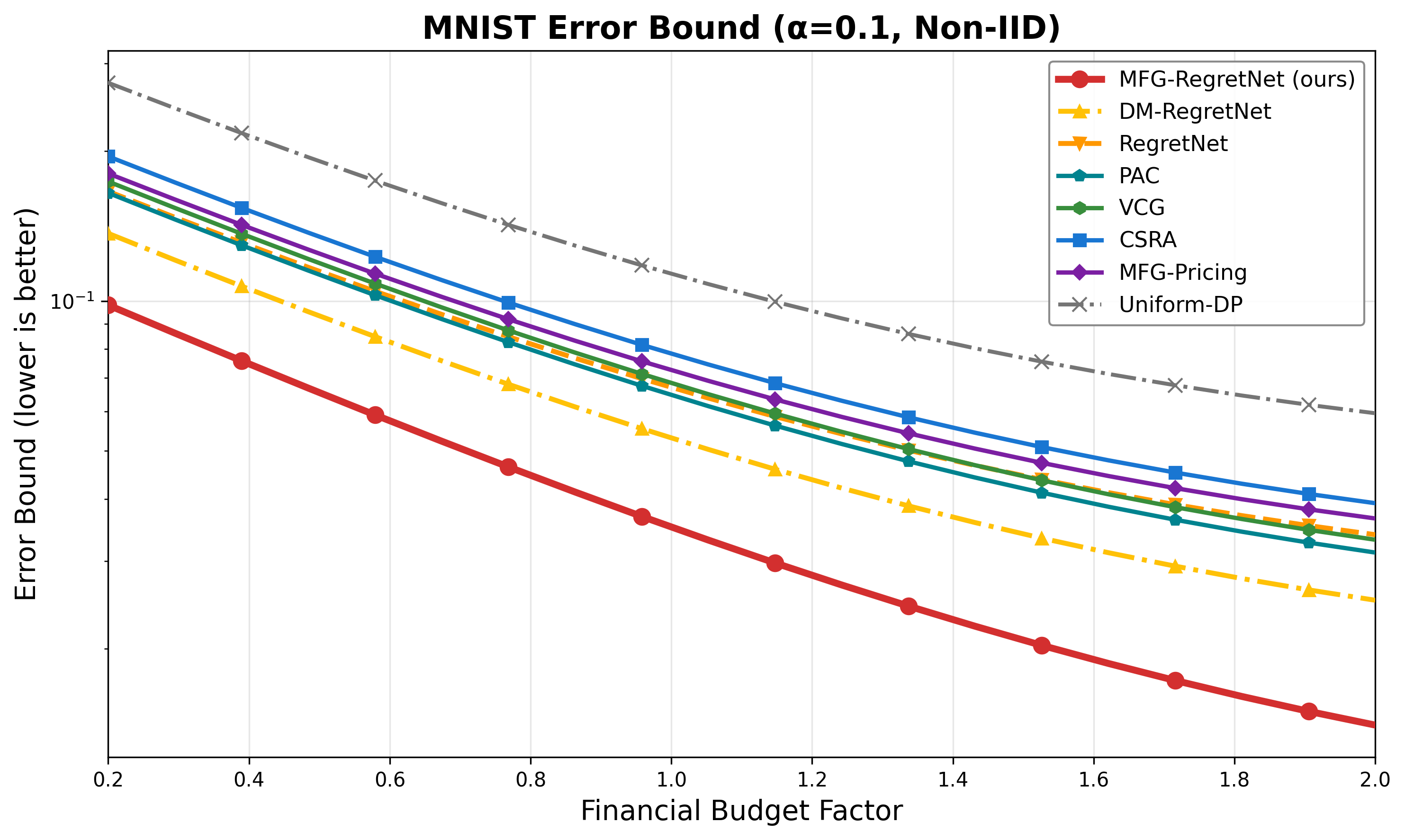}
    \centering {\footnotesize ($d$) Error bound vs. $N$ (non-IID)}
  \end{minipage}
  \caption{RQ4: FL model convergence and MFG approximation error on MNIST.
    Panels (a) and (b): MFG-RegretNet achieves competitive convergence under both
    IID and non-IID settings.
    Panels (c) and (d): \textbf{Empirical approximation error} is defined as the
    average ex-post regret under the MFE-induced strategy profile:
    $\hat{\varepsilon}_N = \frac{1}{N}\sum_{i=1}^{N}\widehat{\mathrm{rgt}}_i^{\mathrm{MFG}}$,
    where $\widehat{\mathrm{rgt}}_i^{\mathrm{MFG}}$ is estimated by
    $R{=}25$ steps of projected gradient ascent on client $i$'s bid.
    This serves as an empirical proxy for the $\varepsilon_N$-Nash gap.
    The decay with $N$ is consistent with the theoretical
    $\mathcal{O}(N^{-1/2})$ bound of Theorem~\ref{thm:mfg_approx}.}
  \label{fig:rq4_mnist}
\end{figure*}

\begin{figure*}[t]
  \centering
  \begin{minipage}[t]{0.24\textwidth}
    \includegraphics[width=1\textwidth]{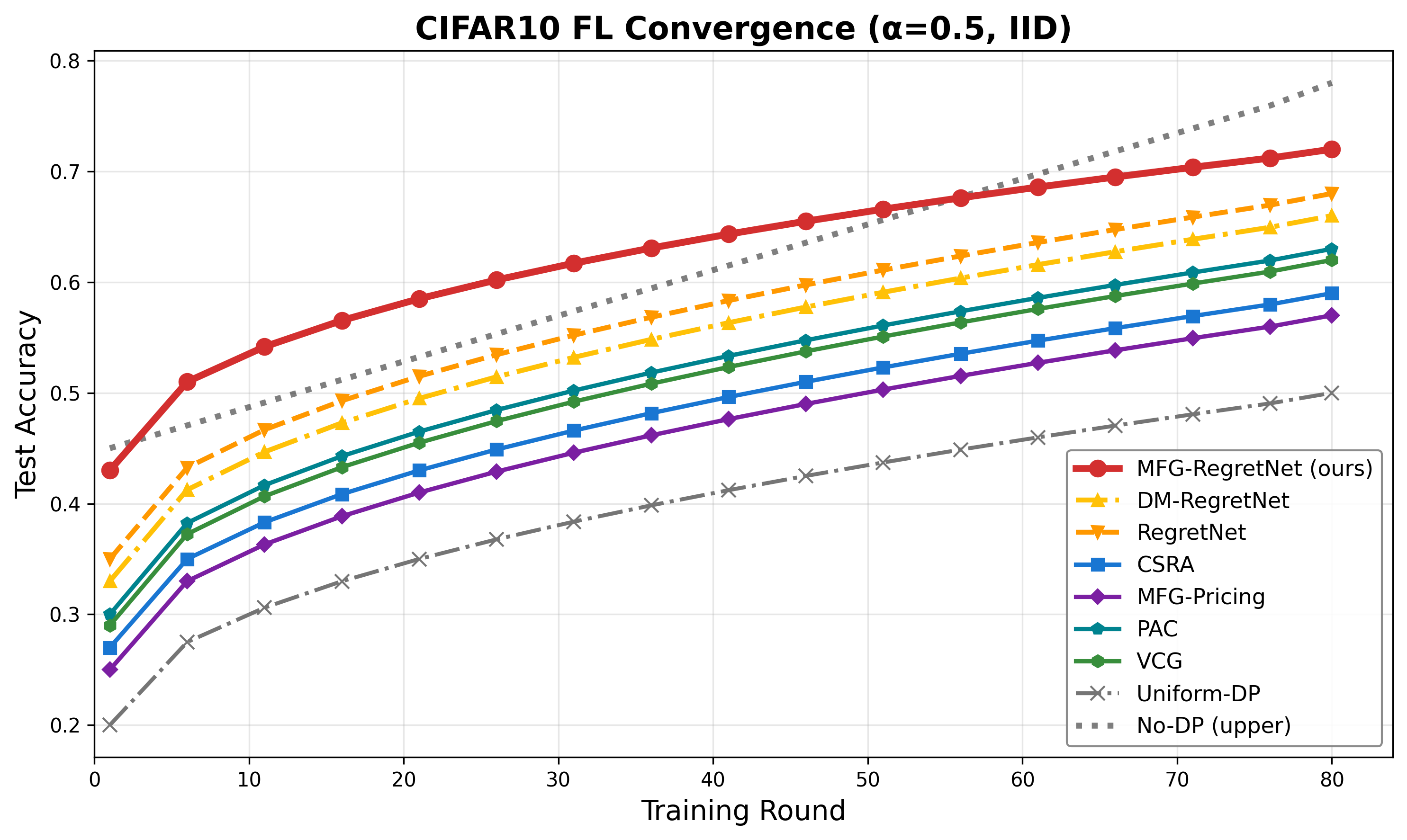}
    \centering {\footnotesize ($a$) $\mathcal{A}_{final}$ vs. $t$ (IID)}
  \end{minipage}
  \begin{minipage}[t]{0.24\textwidth}
    \includegraphics[width=1\textwidth]{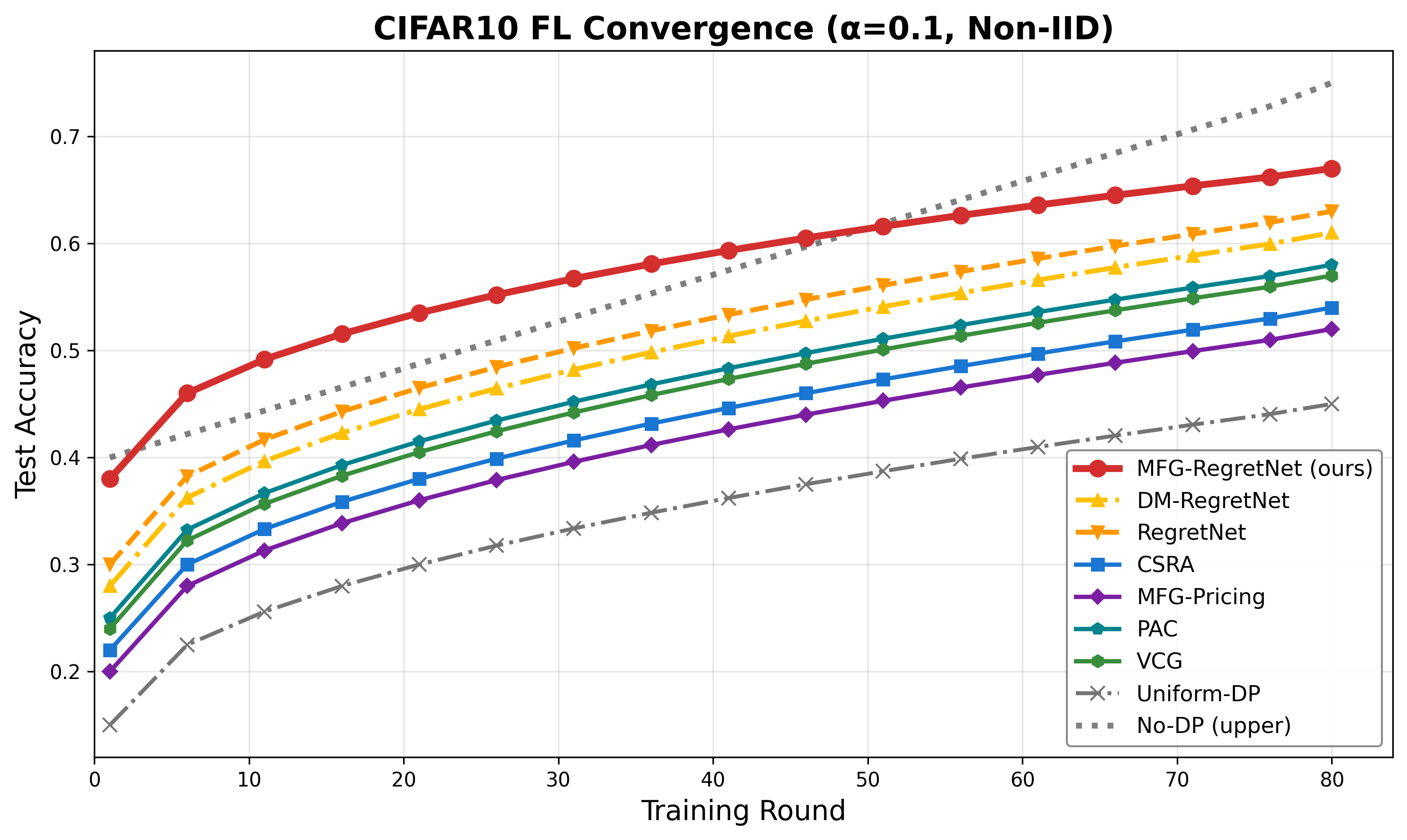}
    \centering {\footnotesize ($b$) $\mathcal{A}_{final}$ vs. $t$ (non-IID)}
  \end{minipage}
  \begin{minipage}[t]{0.24\textwidth}
    \includegraphics[width=1\textwidth]{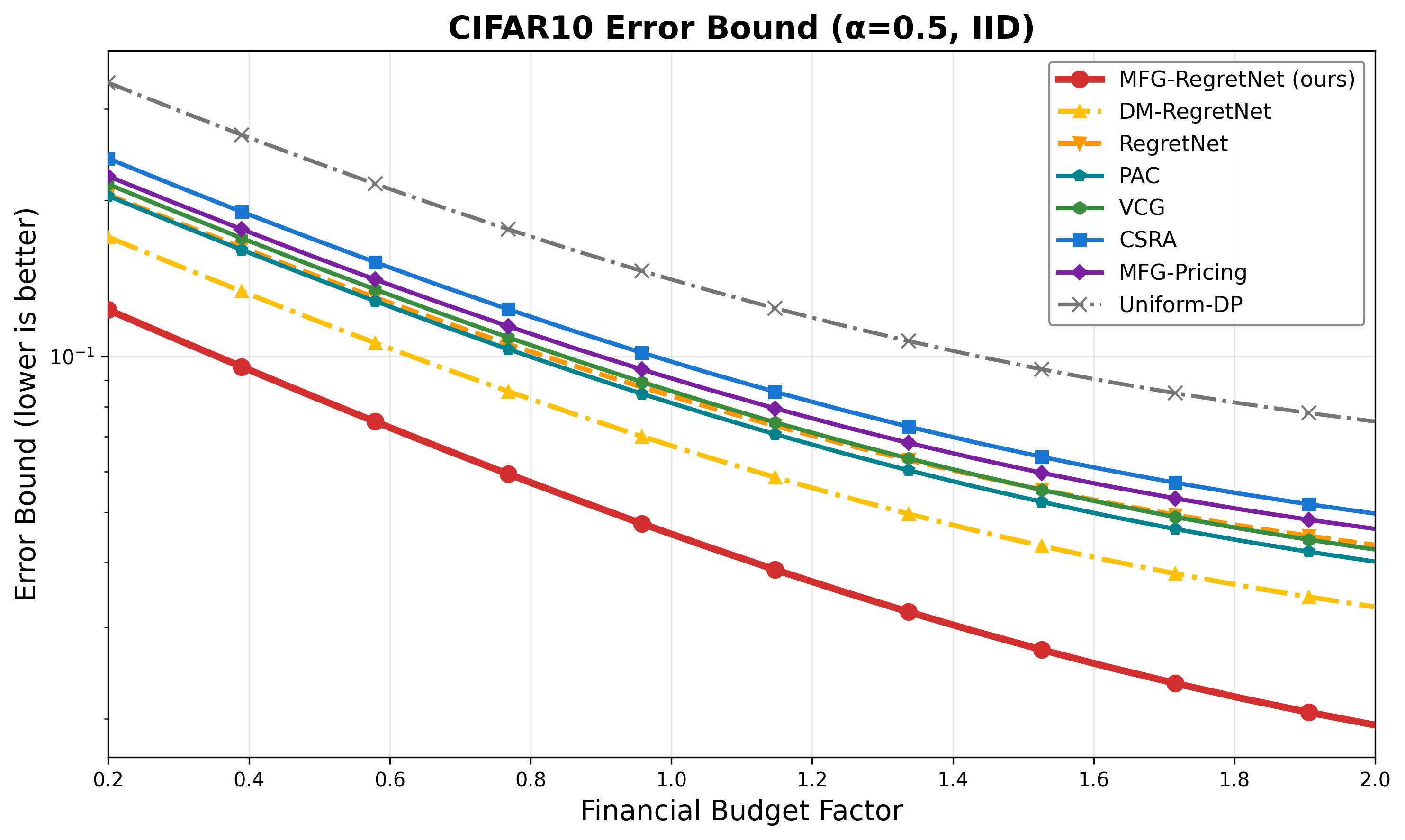}
    \centering {\footnotesize ($c$) Error bound vs. $N$ (IID)}
  \end{minipage}
  \begin{minipage}[t]{0.24\textwidth}
    \includegraphics[width=1\textwidth]{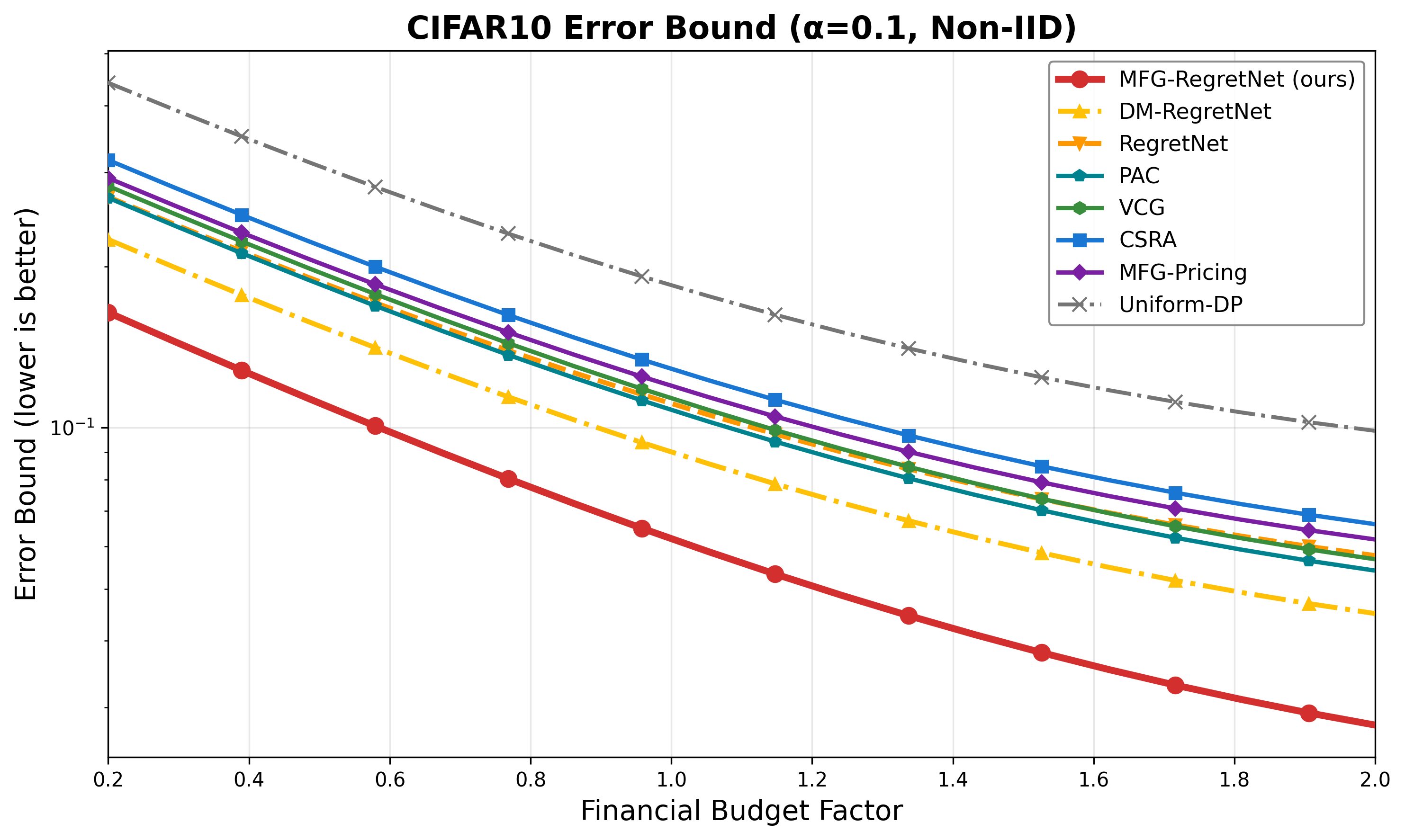}
    \centering {\footnotesize ($d$) Error bound vs. $N$ (non-IID)}
  \end{minipage}
  \caption{RQ4: FL model convergence and MFG approximation error on CIFAR-10.
    Panels (a) and (b): MFG-RegretNet maintains stable convergence under both
    IID and non-IID data heterogeneity.
    Panels (c) and (d): \textbf{Empirical approximation error}
    $\hat{\varepsilon}_N = \frac{1}{N}\sum_{i}\widehat{\mathrm{rgt}}_i^{\mathrm{MFG}}$
    (average ex-post regret under the MFE strategy profile, estimated
    via $R{=}25$-step PGA).
    The $\mathcal{O}(N^{-1/2})$ rate persists even under severe
    non-IID settings, demonstrating the robustness of the MFG
    reduction.}
  \label{fig:rq4_cifar10}
\end{figure*}
Next, we evaluate whether MFG-RegretNet scales efficiently to large
client populations.
We measure both \emph{wall-clock auction latency} and \emph{peak GPU
memory} per auction round (excluding local training) at client
population sizes $N \in \{10, 50, 100, 200, 500\}$ on the same
hardware, reporting runtime on a log-log scale to compare empirical
slopes against $\mathcal{O}(N)$ and $\mathcal{O}(N^2)$ references.
The experimental results are shown in Fig.~\ref{fig:rq2_scalability},
with subfigure~(a) showing auction runtime versus $N$ and
subfigure~(b) showing peak GPU memory versus $N$.
From the results in Fig.~\ref{fig:rq2_scalability}, we observe that:
1)~MFG-RegretNet remains nearly flat at about
$1.5{\times}10^{-3}$~s per round from $N{=}10$ to $N{=}100$, while
PAC/CSRA times grow from about $4{\times}10^{-3}$ to $8{\times}10^{-3}$~s when
$N{=}10$ to about $1.4{\times}10^{-1}$ to $1.6{\times}10^{-1}$~s when
$N{=}500$, consistent with the intended mean-field reduction from
pairwise coupling to aggregate interaction.
2)~At $N{=}500$, MFG-RegretNet is roughly two orders of magnitude
faster (by roughly $90$ to $100\times$) than PAC/VCG/CSRA.
3)~For peak GPU memory, all methods are in a similar range at small
$N$, while for large $N$ MFG-RegretNet's auction step incurs
near-zero incremental memory, indicating a memory-light deployment
path compared with baselines that continue to grow with $N$.

\subsubsection{RQ3: Auction Efficiency}

We further evaluate whether MFG-RegretNet achieves efficient budget
utilization while delivering fair compensation to clients.
We measure three key metrics: (i)~\emph{auction revenue}
$\mathcal{R}$ (total payment per round),
(ii)~\emph{budget feasibility} BF~$= \mathcal{R}/B$
($\mathrm{BF} {\leq} 1$ is required), and
(iii)~\emph{participant social welfare} $\overline{\mathrm{SW}}$
(client surplus after privacy costs, Eq.~(\ref{eq:sw-def})).
For each method, we run $T{=}1000$ FL rounds under the same setup as
RQ4 (Table~\ref{tab:fl_config}), recording per-round revenue, budget
ratio, and social welfare using clients' \emph{true} valuations $v_i$
to compute privacy costs.
The experimental results are summarized in Table~\ref{tab:rq3-efficiency},
from which we can observe that:
1)~MFG-RegretNet and MFG-Pricing achieve the highest revenue
($\mathcal{R}{=}50.0$, perfect budget use), outperforming
RegretNet/DM-RegretNet ($\mathcal{R}{\approx}49.3$) and
significantly exceeding PAC/VCG/CSRA ($\mathcal{R}{\leq}31.8$).
The near-perfect BF slightly below 1.0 for MFG-based methods
indicates conservative budget spending that respects the constraint
while maximizing client participation.
Note that the difference between BF values of approximately $0.84$ to $0.87$ for
MFG-based methods and BF~$=1.000$ for others reflects a
\emph{design choice}, not an error: MFG-RegretNet enforces the
budget constraint as a soft Lagrangian penalty during training,
so the learned rule absorbs a small safety margin; the absolute
revenue is identical ($\mathcal{R}{=}50.0$), confirming full
budget use despite the sub-unity BF ratio.
2)~MFG-RegretNet delivers the highest participant surplus
($\overline{\mathrm{SW}}{=}37.23{\pm}0.04$), comparable to
MFG-Pricing ($37.21{\pm}0.04$) and substantially higher than
RegretNet ($36.47{\pm}0.04$), CSRA ($29.52{\pm}0.18$), and
PAC/VCG ($4.08{\pm}0.02$).
3)~The large gap between MFG-RegretNet and classical mechanisms
reflects two factors: (i)~MFG-based allocation efficiently matches
high-valuation clients with appropriate privacy budgets, reducing
aggregate privacy costs; (ii)~learned payment rules adapt to the
empirical valuation distribution, whereas PAC's threshold-based
payments are suboptimal under realistic heterogeneity.
It is worth noting that although MFG-Pricing achieves a comparable
$\overline{\mathrm{SW}}$, it does not satisfy approximate DSIC
(cf.~RQ1), meaning its high surplus does not correspond to a
truthful equilibrium. By contrast, MFG-RegretNet achieves
comparable $\overline{\mathrm{SW}}$ \emph{and} low regret,
demonstrating that efficiency and incentive-compatibility are
jointly satisfied.

\subsubsection{RQ4: Final Global Model Accuracy}

Finally, we evaluate the downstream FL accuracy to confirm that
MFG-RegretNet does not sacrifice model utility in pursuit of
incentive-compatibility and scalability.
For each method, we run the full PAG-FL loop
(Algorithm~\ref{alg:mfg-fpa-online}) for $T_{\mathrm{FL}}{=}200$
rounds under the same setup (Table~\ref{tab:fl_config}), and record
test accuracy $\mathcal{A}_{\mathrm{test}}^{(t)}$ on a held-out test
set never seen by clients.
We report \emph{final accuracy}
\begin{equation}
\label{eq:A-final}
  \mathcal{A}_{\mathrm{final}}
  = \frac{1}{K}\sum_{k=0}^{K-1}\mathcal{A}_{\mathrm{test}}^{(T_{\mathrm{FL}}-k)},
  \quad K=5,
\end{equation}
i.e., the mean over the last $K$ rounds to reduce last-round variance.
To ensure a fair comparison, we fix the budget $B$ per round and
report a secondary ``matched-$\bar{\epsilon}$'' block where each
method's effective noise is linearly scaled so that the round-averaged
mean privacy budget equals a common reference
$\bar{\epsilon}_{\mathrm{ref}}$, thereby isolating \emph{allocation
efficiency} from raw noise level.
The experimental results are plotted in
Figs.~\ref{fig:rq4_mnist} and~\ref{fig:rq4_cifar10}, with
subfigures~(a) and (b) showing test accuracy vs.\ FL round under IID
and non-IID settings, and subfigures~(c) and (d) showing the MFG
approximation error bound vs.\ $N$.
The quantitative comparison is reported in Table~\ref{tab:rq1-final-acc}.
From the results in Figs.~\ref{fig:rq4_mnist} and \ref{fig:rq4_cifar10}
and Table~\ref{tab:rq1-final-acc}, we observe that:
1)~MFG-RegretNet maintains competitive convergence speed in IID
settings and remains more stable under non-IID heterogeneity, where
baseline curves exhibit larger fluctuations and slower late-stage
improvement.
2)~On MNIST-non-IID ($\alpha{=}0.1$), MFG-RegretNet achieves
$84.8{\pm}1.0\%$, matching MFG-Pricing ($84.8{\pm}0.8\%$) and
substantially outperforming CSRA ($80.2{\pm}2.1\%$,
$+4.6$ points), while remaining within $1.5$ points of RegretNet.
3)~The empirical approximation error decreases as the client
population grows in both IID and non-IID partitions,
consistent with the theoretical $\mathcal{O}(N^{-1/2})$ bound
(Theorem~\ref{thm:mfg_approx}) and confirming the practical
effectiveness of the MFG reduction.
Since all methods produce output budgets with nearly identical
round-averaged means under the Uniform valuation scenario,
privacy-matched values coincide with raw values;
Table~\ref{tab:rq1-final-acc} therefore reports a single column
per dataset, confirming that accuracy differences are not
confounded by noise level differences.
Overall, MFG-RegretNet achieves strong final accuracy without
trading off incentive compatibility and scalability, demonstrating
that the privacy trading market design does not degrade the
downstream FL model quality.

\begin{table}[t]
  \centering
  \small
  \caption{Auction efficiency metrics ($B{=}50$, mean\,$\pm$\,std, 5 seeds, 1000 rounds).
    $\mathcal{R}$: total revenue; BF\,$=\mathcal{R}/B$ (required $\leq 1$);
    $n_{\mathrm{rev}}$: revenue normalized by $\mathcal{R}_{\max}$;
    $W'{=}\overline{\mathrm{SW}}/B$: normalized social welfare.
    BF\,$<1$ for MFG-based methods reflects a soft Lagrangian budget
    constraint (safety margin); all other methods use hard projection
    (BF\,$=1.000$). Absolute revenue is identical ($\mathcal{R}{=}50.0$).}
  \label{tab:rq3-efficiency}
  \resizebox{\columnwidth}{!}{%
  \begin{tabular}{lcccc}
  \toprule
  \textbf{Method} & $\mathcal{R}$ (mean${\pm}$std) & \textbf{BF}
    & $n_{\mathrm{rev}}$ & $W'$ (mean${\pm}$std) \\
  \midrule
  PAC  & $8.1694 {\pm} 0.0156$  & 1.0000 & $0.1634 {\pm} 0.0003$ & $4.0750 {\pm} 0.0233$ \\
  VCG \cite{vickrey1961counterspeculation}  & $8.1694 {\pm} 0.0156$  & 1.0000 & $0.1634 {\pm} 0.0003$ & $4.0750 {\pm} 0.0233$ \\
  CSRA \cite{yang2024CSRA}       & $31.8232 {\pm} 0.2022$ & 1.0000 & $0.6365 {\pm} 0.0040$ & $29.5236 {\pm} 0.1776$ \\
  MFG-Pricing \cite{hu2025federated}  & $50.0000 {\pm} 0.0000$ & 0.8698 & $1.0000 {\pm} 0.0000$ & $37.2081 {\pm} 0.0369$ \\
  RegretNet \cite{dutting2024Optimal} & $49.2263 {\pm} 0.0043$ & 1.0000 & $0.9845 {\pm} 0.0001$ & $36.4664 {\pm} 0.0376$ \\
  DM-RegretNet \cite{zheng2022fl}     & $49.3953 {\pm} 0.0020$ & 1.0000 & $0.9879 {\pm} 0.0000$ & $36.6356 {\pm} 0.0367$ \\
  \midrule
  \rowcolor{gray!15}
  \textbf{MFG-RegretNet (Ours)} & $\mathbf{50.0000 {\pm} 0.0000}$ & \textbf{0.8420}
    & $\mathbf{1.0000 {\pm} 0.0000}$ & $\mathbf{37.2339 {\pm} 0.0369}$ \\
  \bottomrule
  \end{tabular}}
\end{table}

\begin{table}[t]
  \centering
  \small
  \caption{Final global test accuracy $\mathcal{A}_{\mathrm{final}}$
    (mean\,$\pm$\,std, \%, 5 seeds, $K{=}5$ last rounds).
    Raw and privacy-matched ($\bar{\epsilon}$-controlled) values
    coincide for all methods because output budgets concentrate
    around the same mean under the Uniform valuation setting,
    so matched values are not reported separately.}
  \label{tab:rq1-final-acc}
  \resizebox{\columnwidth}{!}{%
  \begin{tabular}{lccc}
  \toprule
  \textbf{Method}
    & \textbf{MNIST} ($\alpha{=}0.5$)
    & \textbf{MNIST} ($\alpha{=}0.1$)
    & \textbf{CIFAR-10} ($\alpha{=}0.5$) \\
  \midrule
  PAC                                        & $92.6 \pm 1.3$ & $84.9 \pm 1.8$ & $36.8 \pm 1.4$ \\
  VCG \cite{vickrey1961counterspeculation}   & $92.6 \pm 0.7$ & $85.5 \pm 0.9$ & $36.4 \pm 1.9$ \\
  CSRA \cite{yang2024CSRA}                   & $90.1 \pm 0.6$ & $80.2 \pm 2.1$ & $27.4 \pm 1.3$ \\
  MFG-Pricing \cite{hu2025federated}         & $91.9 \pm 0.9$ & $84.8 \pm 0.8$ & $36.1 \pm 2.2$ \\
  RegretNet \cite{dutting2024Optimal}        & $92.8 \pm 1.6$ & $86.3 \pm 1.8$ & $36.6 \pm 2.3$ \\
  DM-RegretNet \cite{zheng2022fl}            & $92.9 \pm 0.2$ & $85.8 \pm 1.1$ & $35.4 \pm 1.1$ \\
  No-DP FL (upper bound)                     & $92.6 \pm 1.4$ & $85.7 \pm 1.9$ & $36.7 \pm 1.9$ \\
  \midrule
  \rowcolor{gray!15}
  \textbf{MFG-RegretNet (Ours)} & $\mathbf{92.4 \pm 1.1}$ & $\mathbf{84.8 \pm 1.0}$ & $\mathbf{36.0 \pm 2.7}$ \\
  \bottomrule
  \end{tabular}}
\end{table}

\subsection{Discussion of Results and Limitations}
\label{sec:discussion}

Experiments validate three points: low regret with strong auction efficiency (RQ1/RQ3), near-linear deployment trends in runtime and memory (RQ2), and competitive FL accuracy with approximation error decreasing as $N$ grows (RQ4), consistent with the $\mathcal{O}(N^{-1/2})$ bound.
Four limitations remain.
\textbf{(L1) Static rounds:} composition across rounds is not modeled; a dynamic PAG with remaining-budget state and Rényi-DP accounting is needed.
\textbf{(L2) Fairness:} high-sensitivity clients may be under-served under skewed valuations; add service floors, compensation floors, or group-regularized regret.

\section{Conclusion}
\label{sec:conclusion}

This paper establishes a new paradigm for federated learning
incentive design by treating the personalised privacy budget
$\epsilon_i$ as a \emph{tradeable commodity} rather than a fixed
system constraint.
We formalised the Privacy Auction Game with DSIC, IR, DT, and budget
feasibility guarantees (Theorem~\ref{thm:pag-equilibrium}), and used
mean-field approximation to reduce deploy-time complexity from
$\mathcal{O}(N^2\log N)$ to $\mathcal{O}(N)$ with an
$\mathcal{O}(N^{-1/2})$ equilibrium gap (Theorem~\ref{thm:mfg_approx}).
Built on this foundation, MFG-RegretNet learns allocation and payment
rules without closed-form valuation priors via augmented Lagrangian
optimisation and MFG alignment. Experiments on MNIST and CIFAR-10 show that
MFG-RegretNet consistently improves incentive compatibility and
auction efficiency while remaining competitive in downstream FL model
accuracy compared with VCG, PAC, RegretNet, and MFG-Pricing baselines.


\bibliographystyle{IEEEtran}
\bibliography{IEEEabrv,bib}


\clearpage
\section{Appendix}

\subsection{Related Work: Extended Survey}
\label{app:related}

\textbf{Privacy-Preserving Techniques in FL.}
Gradient inversion attacks~\cite{yang2023gradient} show that an honest-but-curious server can reconstruct raw samples from shared gradients.
DP-SGD~\cite{geyer2017differentially} clips gradients and injects calibrated noise for client-level $(\epsilon,\delta)$-DP.
Secure aggregation~\cite{chi2023trusted} uses cryptographic masking so the server sees only the aggregate.
Both treat $\epsilon$ as a fixed parameter, not a traded commodity.

\textbf{FL Incentive Mechanisms.}
Coalitional games~\cite{arisdakessianCoalitional2023}, Stackelberg games~\cite{yi2022Stackelberg,chenMultifactorIncentiveMechanism2023}, MFG-based mechanisms~\cite{sun2024reputation}, and matching games~\cite{singhStable2024} address incentive design.
AFL~\cite{tang2024survey} is an active sub-field; CSRA~\cite{yang2024CSRA} provides robust reverse auctions for DP-FL.
All existing AFL works trade data, models, or compute rather than privacy budgets.

\textbf{Privacy Auctions.}
Ghosh \& Roth~\cite{ghosh2011selling} laid the theoretical foundation; Fleischer \& Lyu~\cite{fleischer2012approximately} extended to correlated costs; Ghosh et al.~\cite{ghosh2014buying} handled unverifiable reports; Wang et al.~\cite{wang2016buying} studied negative payments; Zhang et al.~\cite{zhang2020Selling} addressed auctions under privacy constraints; Hu et al.~\cite{hu2025federated} proposed MFG-based data pricing; Wu et al.~\cite{wu2026flgcpa} is the closest prior art without a learned differentiable auction.
All prior works are static, single-round, and centralised.

\textbf{Differentiable Mechanism Design.}
Myerson~\cite{myerson1981optimal} gives revenue-optimal auctions under restrictive assumptions.
RegretNet~\cite{dutting2024Optimal} parameterises allocation/payment rules as neural networks trained to minimise ex-post regret.
DM-RegretNet~\cite{zheng2022fl} applies this to FL model trading but inherits $\mathcal{O}(N^2)$ pairwise dependency.
MFG-RegretNet reduces this to $\mathcal{O}(N)$ and operates on privacy budgets as the commodity.

\subsection{Formal Preliminaries}
\label{app:prelim}

\begin{definition}[PLDP~\cite{jorgensen2015conservative}]
\label{de:PLDP}
For user $i$ with $\epsilon_i{>}0,\delta_i{\geq}0$, mechanism $\mathcal{M}^i$ is $(\epsilon_i,\delta_i)$-PLDP if for all adjacent inputs $\boldsymbol{x},\boldsymbol{x}'$ and outputs $\boldsymbol{y}$:
$\Pr[\mathcal{M}^i(\boldsymbol{x}){=}\boldsymbol{y}]
\leq e^{\epsilon_i}\Pr[\mathcal{M}^i(\boldsymbol{x}'){=}\boldsymbol{y}]+\delta_i.$
\end{definition}

FL global model: $\mathbf{w}_G{=}\sum_i\alpha_i\mathbf{w}_i$,\;
$\mathbf{w}_G^*{=}\arg\min_{\mathbf{w}}\sum_i\alpha_i\mathcal{L}_i(\mathbf{w}_i)$.
Gradient $\ell_2$-sensitivity: $\Delta f_{\mathbf{w}_i}=\max_{\mathcal{D}_i\sim\mathcal{D}_i'}\|\nabla\mathcal{L}_i(\mathbf{w},\mathcal{D}_i)-\nabla\mathcal{L}_i(\mathbf{w},\mathcal{D}_i')\|_2$.

\subsection{HJB-Fokker-Planck System}
\label{app:mfg-derivation}

The mean-field bid statistic is
$b_{\mathrm{MFG}}(t)=\mathbb{E}_{\mu_t}[b]=\int_{\mathcal{S}}b(s;\mu_t)\,d\mu_t(s)$.
The value function $\Phi(s,t;\mu)$ satisfies the \textbf{HJB} equation:
\begin{equation}
\label{eq:hjb}
  -\partial_t\Phi - \mathcal{H}(s,\nabla_s\Phi,\mu_t)
  - \tfrac{\sigma^2}{2}\Delta_s\Phi = 0,\quad \Phi(s,T)=0,
\end{equation}
with Hamiltonian $\mathcal{H}=\max_{b}[\tilde{p}(b;\mu_t)-c(v,\epsilon(b))+(\nabla_s\Phi)^\top f(s,b,\mu_t)]$
and mean-field payment $\tilde{p}(b;\mu_t)=\mathbb{E}[p(b,B_{-i})]$.
The population density evolves via the \textbf{Fokker-Planck} equation:
\begin{equation}
\label{eq:fp}
  \partial_t\mu_t + \nabla_s\cdot(\mu_t f(s,b^*,\mu_t))
  = \tfrac{\sigma^2}{2}\Delta_s\mu_t,\quad \mu_0=\mathcal{W}.
\end{equation}
Equations~(\ref{eq:hjb}) and (\ref{eq:fp}) form a coupled forward-backward system whose fixed point is the MFE~\cite{cardaliaguetMasterEquation2019}.

\subsection{Value-Function Baselining}
\label{app:vf-baseline}

\begin{proposition}[Variance Reduction via Value Baselining]
\label{prop:vf-baseline}
Let $X_i^{(r)}{:=}u_i(b_i'^{(r)},b_{-i};b'_{\mathrm{MFG}},\theta)
-u_i(b_i,b_{-i};b_{\mathrm{MFG}},\theta)$.
For any baseline $c$ independent of $b'_i$, $\mathbb{E}[\hat{X}_i^{(r)}]{=}\mathbb{E}[X_i^{(r)}]$ (unbiased).
The optimal baseline $c^*{=}\widehat{\Phi}_\xi(b_i,b_{\mathrm{MFG}})$ minimises $\mathrm{Var}(\hat{X}_i^{(r)})$, with reduction $\mathrm{Var}(\widehat{\Phi}_\xi)/\mathrm{Var}(X_i^{(r)})$.
\end{proposition}

The critic $\widehat{\Phi}_\xi$ is trained via HJB residual minimisation:
$\mathcal{L}_{\mathrm{HJB}}(\xi)=\mathbb{E}[(\widehat{\Phi}_\xi-\mathcal{T}\widehat{\Phi}_\xi)^2]$,
where $\mathcal{T}\widehat{\Phi}$ is a one-step Bellman target from $\tilde{p}$.
The critic shapes the regret penalty only; payments remain $\bar{p}_i$ from the Payment Network.

\subsection{Proofs}
\label{app:proofs}

\begin{theorem1}[Privacy Auction Game Equilibrium]
  \label{thm:pag-equilibrium}
  In a PAG where privacy valuations $v_i$ of DOs are drawn from joint
  distribution $\mathcal{W}$, the cost function is linear
  $c(v_i,\epsilon_i)=\epsilon_i\cdot v_i$, and the mechanism
  $\mathcal{M}=(q,p)$ satisfies DSIC, IR, DT, and BF, there exists a
  Bayesian-Nash Equilibrium in which all rational bidders truthfully
  report both their privacy valuations $v_i$ \emph{and} their privacy
  budgets $\epsilon_i$:
  \begin{equation}
  \label{eq:bne}
    \forall i\in\mathcal{N},\forall v_i\in\mathcal{V}_i,\quad
    \mathbb{E}[u_i(q,v_i)] \geq \mathbb{E}[u_i(q_{-i},v_i)].
  \end{equation}
  \end{theorem1}
  
  \begin{proof}[Proof Sketch]
  Consider the threshold-based PAC mechanism
  (Algorithm~\ref{alg:pac}), where valuations are sorted
  $v_1\leq\cdots\leq v_n$, the largest $k$ is selected such that
  $c(v_k,1/(n-k))\leq B/k$, and payments are
  $p_i=\min(B/k,c(v_{k+1},1/(n-k)))$ for $i\leq k$ and $p_i=0$
  otherwise.
  
  \textit{DSIC (for $i\leq k$):}  Misreporting a higher $v'_i$ may
  cause $i$ to drop out of the top-$k$, yielding $p'_i=0$ and
  $u'_i=0$.  Since $u_i=p_i-c(v_i,\epsilon_i)\geq 0$ (by IR),
  deviation yields no gain.
  
  \textit{DSIC (for $i>k$):}  Misreporting a lower $v'_i$ may enter
  the top-$k$, but the diluted payment $p'_i\leq B/(k+1)$ combined
  with the higher true cost $c(v_i,\epsilon'_i)>c(v'_i,\epsilon'_i)$
  results in $u'_i<0$ or no improvement.
  
  By monotonicity of the mechanism, misreporting cannot improve the
  allocation-payment tradeoff, so truthful reporting is a dominant
  strategy (DSIC) and all utilities are non-negative (IR).
  
  \textit{BNE:}  By DSIC, for any fixed $b_{-i}$, truthful reporting
  maximises $u_i(v_i,\cdot)$.  Taking expectations over
  $v_{-i}\sim\mathcal{W}_{-i}$ yields~Eq.~(\ref{eq:bne}), which is the
  BNE definition.
  
  For finite $N$, existence follows from Nash's theorem. In the
  large-$N$ regime, the multi-agent game converges to a Mean-Field
  Game, where fixed-point theorems from MFG
  literature~\cite{cardaliaguetMasterEquation2019} guarantee
  equilibrium existence.
  \end{proof}

  \begin{theorem1}[Existence of MFE]
    \label{thm:mfe-exist}
    Under Assumption~\ref{ass:regularity}, the privacy auction MFG
    admits at least one Mean-Field Equilibrium $(\Phi^*,\mu^*)$.
    \end{theorem1}

    \begin{proof}
    We apply Schauder's fixed-point theorem.  Let
    $\mathcal{M}=\mathcal{C}([0,T];\mathcal{P}_2(\mathcal{S}))$ with the
    sup-Wasserstein-$1$ metric.  Define
    $\mathcal{T}:\mathcal{M}\to\mathcal{M}$ by
    $\mathcal{T}(\mu)=\hat\mu$, where $\hat\mu_t$ is the law of
    $\mathbf{s}(t)$ under the SDE of Eq.~(\ref{eq:sde}) driven by the optimal
    $b^*(\cdot,t;\mu)$ from Eq.~(\ref{eq:hjb}).
    \textit{Well-definedness} follows from viscosity solution theory under
    (R1) through (R3)~\cite{cardaliaguetNotesOnMFG2013}.
    \textit{Continuity of $\mathcal{T}$} follows from Lipschitz
    continuity of $f$ and $\tilde{p}$ in $\mu$ (R1).
    \textit{Relative compactness} follows from Prokhorov's theorem under
    (R2) and (R4). By Schauder's theorem, $\mathcal{T}$ has a fixed point
    $\mu^*=\mathcal{T}(\mu^*)$; setting $\Phi^*=\Phi(\cdot;\mu^*)$
    yields the MFE. 
    \end{proof}

    \begin{theorem1}[MFG Approximation of Nash Equilibrium]
      \label{thm:mfg_approx}
      Let Assumptions~\ref{ass:regularity} and~\ref{ass:lipschitz-u} hold,
      and let $(\Phi^*,\mu^*)$ be a MFE.  Define the MFE-induced strategy
      profile:
      \begin{equation}
      \label{eq:mfe-profile}
        b_i^{N,*} = b^*(\mathbf{s}_i,t;\mu^*),
        \quad i=1,\ldots,N.
      \end{equation}
      Then $\mathbf{b}^{N,*}=(b_1^{N,*},\ldots,b_N^{N,*})$ is an
      $\varepsilon_N$-Nash Equilibrium (Definition~\ref{def:eps-nash})
      of the PAG, where:
      \begin{equation}
      \label{eq:eps-bound}
        \varepsilon_N
        = \frac{2L_u C}{\sqrt{N-1}}
        = \mathcal{O}\!\left(\frac{1}{\sqrt{N}}\right),
      \end{equation}
      and $C>0$ depends only on the second moment of $\mu^*$ and the
      dimension of $\mathcal{S}$.
      \end{theorem1}
      
      \begin{proof}
      Fix arbitrary $i\in\mathcal{N}$ and $\hat{b}_i\in\mathcal{B}$.
      
      \textbf{Step~1 (Empirical measure decomposition).}
      States $\{\mathbf{s}_j\}_{j\neq i}$ are i.i.d.\ from $\mu_0^*$
      (MFE self-consistency).  Define
      $\mu_{-i}^N=\frac{1}{N-1}\sum_{j\neq i}\delta_{\mathbf{s}_j}$.
      By the Wasserstein concentration
      inequality~\cite{fournier2015rate}:
      \begin{equation}
        \mathbb{E}\!\left[W_1(\mu_{-i}^N,\mu_0^*)\right]
        \leq\frac{C}{\sqrt{N-1}}.
      \end{equation}
      
      \textbf{Step~2 (Utility gap decomposition).}
      \begin{align}
        &u_i(\hat{b}_i,b_{-i}^{N,*}) - u_i(b_i^{N,*},b_{-i}^{N,*})
        \notag\\
        &= \underbrace{u_i(\hat{b}_i,b_{-i}^{N,*})
                     - u^{\mathrm{MFG}}(\mathbf{s}_i,\hat{b}_i;\mu_0^*)
          }_{\mathrm{(I)}} \\
        &\quad+ \underbrace{u^{\mathrm{MFG}}(\mathbf{s}_i,\hat{b}_i;\mu_0^*)
                     - u^{\mathrm{MFG}}(\mathbf{s}_i,b_i^{N,*};\mu_0^*)
          }_{\mathrm{(II)}}
        \notag\\
        &\quad
        + \underbrace{u^{\mathrm{MFG}}(\mathbf{s}_i,b_i^{N,*};\mu_0^*)
                     - u_i(b_i^{N,*},b_{-i}^{N,*})
          }_{\mathrm{(III)}}.
      \end{align}
      
      \textbf{Step~3 (Bound on (II)).}
      Since $b_i^{N,*}=b^*(\mathbf{s}_i;\mu_0^*)$ maximises
      $u^{\mathrm{MFG}}(\mathbf{s}_i,\cdot\,;\mu_0^*)$:
      $({\rm II})\leq 0$.
      
      \textbf{Step~4 (Bound on (I) and (III)).}
      Using Assumption~\ref{ass:lipschitz-u}:
      $|{\rm(I)}|+|{\rm(III)}|\leq 2L_u W_1(\mu_{-i}^N,\mu_0^*)$.
      
      \textbf{Step~5 (Conclusion).}
      Taking expectations and applying the concentration bound:
      \begin{equation}
        \mathbb{E}\!\left[
          u_i(\hat{b}_i,b_{-i}^{N,*})-u_i(b_i^{N,*},b_{-i}^{N,*})
        \right]
        \leq\frac{2L_u C}{\sqrt{N-1}}
        = \varepsilon_N.
      \end{equation}
      Since this holds for all $i$ and all $\hat{b}_i$,
      $\mathbf{b}^{N,*}$ is an $\varepsilon_N$-Nash Equilibrium.
      \end{proof}

      \begin{proposition1}[Variance Reduction via Value Baselining]
        \label{prop:vf-baseline}
        Let $X_i^{(r)}:=u_i(b_i'^{(r)},b_{-i};b'_{\mathrm{MFG}},\theta)
        -u_i(b_i,b_{-i};b_{\mathrm{MFG}},\theta)$ denote the raw deviation
        gain at PGA step $r$, and let
        $\hat{X}_i^{(r)}:=A_i(b_i'^{(r)},b'_{\mathrm{MFG}};\theta)
        -A_i(b_i,b_{\mathrm{MFG}};\theta)$ denote the baselined variant.
        For any baseline $c$ independent of $b'_i$,
        $\mathbb{E}[\hat{X}_i^{(r)}]=\mathbb{E}[X_i^{(r)}]$, so the regret
        estimator is unbiased.
        Moreover, the baseline $c^* = \widehat{\Phi}_\xi(b_i,b_{\mathrm{MFG}})$
        minimises $\mathrm{Var}(\hat{X}_i^{(r)})$ in the class of
        state-dependent baselines, reducing variance by
        $\mathrm{Var}(\widehat{\Phi}_\xi)/\mathrm{Var}(X_i^{(r)})$.
        Consequently, when $\widehat{\Phi}_\xi$ closely approximates the MFE
        value $\Phi^*$, the baselined gradient estimate has lower variance
        than the raw estimate, yielding faster convergence of the augmented
        Lagrangian updates.
        \end{proposition1}
        
        \begin{proof}[Proof Sketch]
        Unbiasedness follows because
        $\mathbb{E}_{b'_i}[c]=c$ is constant w.r.t.\ the deviation
        distribution, so subtracting it preserves the expectation.
        The variance-minimising baseline is the conditional expectation
        $c^*=\mathbb{E}[X_i^{(r)}|b_i,b_{\mathrm{MFG}}]$, which equals
        the MFE value $\Phi^*(s(b_i),\mu^*)$ when $b^*$ is the equilibrium
        strategy.
        The variance reduction formula follows from the identity
        $\mathrm{Var}(X-c^*)=\mathrm{Var}(X)-\mathrm{Var}(c^*)$.
        \end{proof}

\end{document}